\documentclass[twoside,leqno,twocolumn]{article}

\usepackage[letterpaper]{geometry}

\usepackage{ltexpprt}
\usepackage{hyperref}

\usepackage[many]{tcolorbox}    

\newtcolorbox{mybox}[1]{%
    tikznode boxed title,
    enhanced,
    arc=0mm,
    interior style={white},
    attach boxed title to top center= {yshift=-\tcboxedtitleheight/2},
    fonttitle=\bfseries,
    colbacktitle=white,coltitle=black,
    boxed title style={size=normal,colframe=white,boxrule=0pt},
    title={#1}}
    
\begin{document}

\newcommand\relatedversion{}
\renewcommand\relatedversion{\thanks{The full version of the paper can be accessed at \protect\url{https://arxiv.org/abs/1902.09310}}} 

\title{\Large A Framework for Computing Greedy Clique Cover}
\author{Ahammed Ullah\thanks{Purdue University}
}

\maketitle







\begin{abstract} \small\baselineskip=9pt 
Structural parameters of graph (such as degeneracy and arboricity) had rarely been considered when designing algorithms for \emph{(edge) clique cover} problems. Taking degeneracy of graph into account, we present a greedy framework and two fixed-parameter tractable algorithms for \emph{clique cover} problems. We introduce a set theoretic concept and demonstrate its use in the computations of different objectives of \emph{clique cover}. Furthermore, we show efficacy of our algorithms in practice.
\end{abstract}

\section{Introduction}
\label{sec:Intro}

A set of cliques $\C$ is an \emph{(edge) clique cover} of an undirected graph if every edge of the graph is contained in at least one of the cliques in $\C$. Computing a \emph{clique cover} with smallest number of cliques is an NP-hard problem \cite{orlin1977contentment}, and the problem is constant-factor inapproximable unless $P = NP$ \cite{lund1994hardness}. \emph{Clique cover} problems abstract an abundance of real-world problems, arising in resource scheduling \cite{roberts1985applications}, causal structure learning \cite{markham2020measurement}, stringology \cite{helling2018constructing}, computational biology [\cite{blanchette2012clique}, \cite{cooley2021parameterized}], computational geometry \cite{agarwal1994can}, compiler optimization \cite{rajagopalan2000handling}, and applied statistics [\cite{ennis2012assignment}, \cite{gramm2007algorithms}, \cite{piepho2004algorithm}].

In parameterized complexity, a parameterized problem $L \subseteq \sum^{*} \times \mathbb{N}$ is called fixed-parameter tractable (FPT), if any $(x,p) \in L$ is decidable in time $f(p).|(x,p)|^{O(1)}$, where $f(p)$ is computable and is independent of $x$, but can grow arbitrarily with $p$. An algorithm is called fixed-parameter tractable algorithm, if it decides $(x,p) \in L$ in time bounded by $f(p).|(x,p)|^{O(1)}$. A parameterized decision problem of \emph{clique cover} is as follows.

\begin{mybox}{ECC}
\textbf{Input:} An undirected graph $G=(V,E)$, and a nonnegative integer $k$.

\textbf{Output:} If exists, a set of at most $k$ cliques of $G$ such that every edge of $G$ is contained in at least one of the cliques in the set; otherwise 'no'.
\end{mybox}

\emph{ECC} is fixed-parameter tractable with respect to parameter $k$ \cite{gramm2006data}. Using polynomial time data reduction rules, \cite{gramm2006data} have shown a kernel with $2^k$ vertices, and a search tree algorithm with time bound double exponential in $k$. Unless the polynomial hierarchy collapses, no polynomial kernel exists for \emph{ECC}. Moreover, double exponential dependence on $k$ is necessary, assuming the exponential time hypothesis.

\textcolor{red}{The FPT results are substantially extended and included in a different article, see} \url{http://arxiv.org/abs/2208.12438}\textcolor{red}{. This article will be updated to contain the greedy results.}

In general, $k$ (clique cover size) is large [\cite{alon1986covering}, \cite{erdos1966representation}, \cite{frieze1995covering}, \cite{gyarfas1990simple}, \cite{pullman1983clique}, \cite{roberts1985applications}], specially for sparse graphs. The results on lower bound of $k$, and the results of non-existence of polynomial kernel for \emph{ECC} with respect to clique cover size imply limited scope of designing tractable FPT algorithms with $k$ as only parameter. Therefore, it is worthy exploring different structural parameters of graph, such as degeneracy (a sparsity measure of graph), as additional parameters along with $k$.

\emph{Degeneracy} of a graph $G=(V,E)$ is the smallest value $d$ such that every nonempty subgraph of $G$ contains a vertex that has at most $d$ adjacent vertices in the subgraph. \emph{Degeneracy ordering} of $G$ is an ordering of vertices in $V$ such that each vertex has at most $d$ neighbors that come later in the ordering. Degeneracy ordering is computable in linear time \cite{matula1983smallest}.

\emph{Degeneracy} is within a constant factors of other popular sparsity measures of graph such as arboricity \cite{chiba1985arboricity} and thickness \cite{dean1991thickness}, and is also at most equal to the treewidth \cite{robertson1984graph} of graph; however, a graph may have bounded degeneracy, but unbounded treewidth, for example, grid graphs. Degeneracy is also known as the $d$-core number \cite{seidman1983network}, width \cite{freuder1982sufficient}, linkage \cite{kirousis1996linkage}, and is equivalent to the coloring number \cite{erdHos1966chromatic}.

Throughout the paper we use following notations. For a graph $G= (V,E)$, $n = |V|$ is the number of vertices, $m = |E|$ is the number of edges, $N(x)$ denote the set of neighbours of a vertex $x \in V$, $\Delta$ is maximum degree of the vertices, and $d$ is the degeneracy of $G$. Let $u_1, u_2, ..., u_n$ be a degeneracy ordering of $V$. $N_d(u_i)$ denotes neighbours of $u_i$ that follow $u_i$ in the degeneracy order of $V$, i.e., $N_d(u_i) = N(u_i) \cap \{u_{i+1}, ..., u_n\}$. The definition of fixed-parameter tractability encompasses more general cases such as when $p$ is a vector of nonnegative integers. Let \emph{ECC-d} be the variant of the parameterized decision problem \emph{ECC} that considers degeneracy of $G$ explicitly as parameter along with $k$. Suppressing factors polynomial in input size, we express time bound of an FPT algorithm with $O^{*}(f(p))$.

\subsection{Our Contributions.}

We introduce a set theoretic concept (we call it \emph{candidate clique sets}) closely related to family of sets representing vertices of an intersection graph \cite{mckee1999topics} and use it to design greedy algorithms and FPT algorithms.

We define a greedily computable \emph{clique cover} and show that it admits desirable characterizations. Combining \emph{candidate clique sets} with different objectives of \emph{clique cover}, we derive characterizations useful in designing and analyzing algorithms.

Using our derived characterizations, we describe a basic greedy framework that can compute a \emph{minimal} (inclusion-wise) \emph{clique cover} in $O(m^2)$ time.

We make use of \emph{candidate clique sets} and related characterizations to improve the time bound of our greedy framework; for graphs of constant maximum degree, this results in a family of linear time algorithms. We make the improved greedy framework amenable to utilize degeneracy of input graph; this also results in a family of linear time algorithms, for graphs of constant degeneracy.

We present two FPT algorithms for \emph{ECC-d} with respect to parameter $(d,k)$. One of our FPT algorithm is based on \emph{candidate clique sets}, and the other FPT algorithm is based on maximal clique enumeration (listing all maximal cliques in graph \cite{eppstein2013listing}). For graphs of constant degeneracy, time bounds of our FPT algorithms have single exponential dependence on $k$.

For graphs satisfying $2^k > d$, our maximal clique enumeration based FPT algorithm improves upon the time bound of the search tree algorithm of \cite{gramm2006data} by a factor of $3^{k(2^k-d)}$.

Our greedy framework and FPT algorithm based on \emph{candidate clique sets} demonstrate a general design paradigm for different objectives of \emph{clique covers}. For example, we easily obtain a search tree algorithm for the \emph{assignment-minimum clique cover} (\emph{clique cover} with a minimum total number of vertices across all cliques) [\cite{ennis2012assignment}, \cite{gramm2007algorithms}, \cite{piepho2004algorithm}] with depth, number of branches, and space use all bounded by polynomials in input size.
Prior search tree algorithm designed for this objective of \emph{clique cover} requires bounds exponential in input size for depth, number of branches, and space use \cite{ennis2012assignment}.

Our algorithms show excellent performance in practice. We show comparison results of different metrics of our algorithms against the state of the art, on random graphs, as well as on real-world graphs.
\section{Related Works}
\label{sec:related}

Computing \emph{clique cover} with a minimum number of cliques is NP-hard \cite{orlin1977contentment}, even for planar graphs \cite{chang2001tree}, or graphs with maximum degree six \cite{hoover1992complexity}, whereas the problem is solvable in polynomial time for some restricted classes of graphs such as graphs with maximum degree five \cite{hoover1992complexity}, chordal graphs \cite{ma1989clique}, line graphs \cite{orlin1977contentment}, certain generalizations of line graphs \cite{prisner1995clique}, and circular-arc graphs \cite{hsu1991linear}. The minimum cardinality \emph{clique cover} problem is not approximable within a factor of $n^\epsilon$, for some $\epsilon > 0$ unless $P = NP$ \cite{lund1994hardness}; nothing better than a polynomial approximation ratio of $O(n^2 \frac{( \log \log n)^2}{(\log n)^3} )$ is known \cite{ausiello2012complexity}. Unlike \emph{vertex clique cover} (set of disjoint cliques covering vertices of a graph; equivalent to coloring the complement graph), \emph{edge clique cover} is fixed-parameter tractable with respect to natural parameterization \cite{gramm2006data}. \cite{blanchette2012clique} presented an FPT algorithm for \emph{clique cover} parameterized by treewidth. For planar graphs, \cite{blanchette2012clique} also presented an FPT algorithm using branchwidth as parameter and a polynomial-time approximation schemes. Several other objectives of \emph{clique cover} also have been shown to be FPT including \emph{clique cover} objectives that consider edge weights \cite{cooley2021parameterized, ullah2022computing}.

Kellerman presented a heuristic algorithm (Kellerman heuristic) to solve keyword conflict problem \cite{kellerman1973determination}, equivalent to \emph{clique cover} \cite{kou1978covering}. \cite{kou1978covering} proposed an extension to the Kellerman heuristic. Including the extension, the time bound of the Kellerman heuristic is $O(nm^2)$ \cite{gramm2006data}. Using additional data structures, \cite{gramm2006data} improved the time bound to $O(nm)$. \cite{conte2020large} presented a general framework of heuristic algorithms and evaluated the framework for different objectives of \emph{clique cover}. Variants of algorithms from the framework of \cite{conte2020large} have $O(m\Delta)$ or $O(md)$ time bound. Improvement of Kellerman heuristic \cite{gramm2006data} or heuristics proposed by \cite{conte2020large} outperform other heuristic algorithms [\cite{gramm2007algorithms}, \cite{helling2018constructing}, \cite{piepho2004algorithm}] developed for the problem. \cite{behrisch2006efficiently} designed a polynomial time algorithm that is asymptotically optimal for certain random intersection graphs, but for arbitrary graphs, this algorithm may not compute a \emph{clique cover}.
\section{Key Concepts}
\label{sec:kernel}
In this section we present a greedily computable \emph{clique cover} and its characterizations. We introduce a set theoretic concept (\emph{candidate clique sets}) closely related to family of sets representing vertices of an intersection graph. For the computations of different objectives of \emph{clique cover}, we derive useful characterizations of \emph{candidate clique sets}.

\begin{figure}
    \centering
     \tikzset{main node/.style={circle,fill=black!20,draw,minimum size=0.8cm,inner sep=0pt},
            }
 \begin{tikzpicture}
    \node[main node] (1) {$1$};
    \node[main node] (5) [above right = 0.6cm and 1.1cm of 1]  {$5$};
    \node[main node] (3) [below right = 0.5cm and 1.1cm of 1] {$3$};
    \node[main node] (6) [right = 1.5cm of 3] {$6$};
    \node[main node] (2) [below right = 0.9cm and 0.8cm of 3] {$2$};
    \node[main node] (4) [below left = 0.9cm and 0.8cm of 3] {$4$};

    \path[draw,thick]
    (1) edge node {} (5)
    (1) edge node {} (3)
    (2) edge node {} (3)
    (2) edge node {} (4)
    (3) edge node {} (4)
    (3) edge node {} (5)
    (3) edge node {} (6)
    (4) edge node {} (5)
    (4) edge node {} (6)
    (5) edge node {} (6)
    ;
\end{tikzpicture} 
    \caption{A graph with six vertices and ten edges.}
    \label{fig:example_graph}
\end{figure}
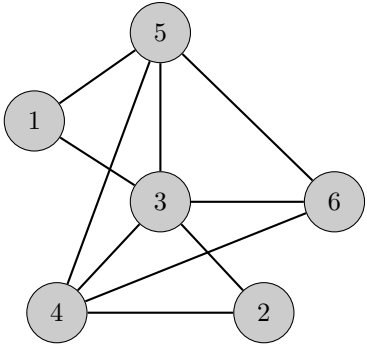

\begin{Definition}[Minimum Clique Cover]
A \emph{minimum clique cover} of a graph $G$ is a \emph{clique cover} with smallest number of cliques among all \emph{clique covers} of $G$.
\end{Definition}

For the example graph of Figure \ref{fig:example_graph}, a \emph{minimum clique cover} is $\C = \{\{1, 3, 5\}, \{2, 3, 4\}, \{3, 4, 5, 6\}\}$. In this case all the cliques of the \emph{minimum clique cover} are maximal cliques. But having a \emph{clique cover} consisting only of maximal cliques is neither necessary nor sufficient for a \emph{minimum clique cover}.

\begin{Definition}[Minimal Clique Cover]
A \emph{minimal clique cover} $\C$ of a graph $G$ is a \emph{clique cover} of $G$ such that no other \emph{clique cover} of $G$ is contained in $\C$, i.e., $\C$ is an inclusion-wise minimal set of cliques that covers all the edges of $G$.
\end{Definition}

A \emph{minimum clique cover} is also a \emph{minimal clique cover}. Note that from a \emph{minimal clique cover} it may be possible to remove vertices (or edges), given that such removals still maintain a \emph{clique cover} of $G$. Any such removal of vertices (or edges) from a \emph{minimal clique cover} would not affect number of cliques in the \emph{clique cover}.

In the example graph of Figure 1, $\C = \{\{1,3,5\}, \{2,3,4\}, \{4,5,6\}, \{3,6,5\}\}$ is a \emph{minimal clique cover} of the graph, and size of the \emph{clique cover} is four. From the clique $\{3,6,5\}$, we can remove vertex 5 (hence edges $\{3,5\}$ and $\{6,5\}$), but the size of the \emph{clique cover} would still be four.

It is easy to construct a \emph{minimal clique cover} by covering every edge with a clique containing only the edge. But, computing a \emph{minimal clique cover} with some desirable characterizations does not seem to be straightforward. Therefore, we define a type of \emph{clique cover} that admits desirable characterizations and can be efficiently turned into a \emph{minimal clique cover}, preserving the characterizations.

\begin{Definition}[Locally Minimal Clique Cover]
\label{local_min}
Let $\C = \{C_1, C_2, ..., C_k\}$ be a \emph{clique cover} of a subset of edges of $G=(V,E)$, and $\{x,y\}$ be an uncovered edge of $G$. Let $\C^{'}$ be a clique cover of the edges covered in $\C$ together with the edge $\{x, y\}$ constructed as follows.

(i) If edge $\{x,y\}$ can be covered by a clique contained in $\C$, then we set $\C^{'}$ to be a clique cover in which $\{x,y\}$ is covered with exactly one of the cliques of $\C$ and rest of the cliques of $\C$ are included.

(ii) Otherwise, we create a new clique $C_{k+1}$ containing only $\{x,y\}$, and we set $\C^{'}$ to be $\C \cup \{C_{k+1}\}$.

We call a \emph{clique cover} $\C$ \emph{locally minimal} if $\C$ is obtained from an empty clique cover using aforesaid construction; i.e., for every expansion of the cliques contained in $\C$, we either have used (i) whenever applicable, or (ii) otherwise. 
\end{Definition}

Note that a \emph{locally minimal clique cover} only ensures local minimality with respect to an uncovered edge of interest; it does not ensure minimality with respect to the edges already included in the clique cover. Nevertheless, it allows characterizations of \emph{clique covers} desirable for designing and analyzing algorithms. 

\begin{restatable}{lemma}{LemmaCharOne}
\label{lemma_lm_char1}
Let $\C = \{C_1, C_2, ..., C_k\}$ be a clique cover of a subset of edges of $G=(V,E)$. If $\C$ is \emph{locally minimal}, then each vertex $x$ appears at most $|N(x)|$ times in $\C$.
\end{restatable}
\begin{proof}
Since $\C$ is \emph{locally minimal}, for an uncovered edge $\{x,y\}$ covered during the expansion of clique cover at hand, we can trace use of the two rules described in Definition \ref{local_min}. Let $C^{'}$ be the clique cover immediately before covering the edge $\{x,y\}$. Before the expansion of $C^{'}$ for $\{x,y\}$, none of the cliques in $C^{'}$ contained both $x$ and $y$: if $x$ and $y$ are both contained in a clique of $C^{'}$, then clearly $\{x,y\}$ is not uncovered. Number of uncovered edges $\{x,y\}$ incident on a vertex $x$ is at most $|N(x)|$. Therefore, number of such expansions of a \emph{clique cover} for a vertex $x$ is at most $|N(x)|$. Hence the vertex $x$ appears at most $|N(x)|$ times in $\C$.
\end{proof}

\begin{Remark}
\label{rem_2}
The converse of characterization in Lemma \ref{lemma_lm_char1} does not hold; i.e., each $x \in V$ can appear at most $|N(x)|$ times in $\C$, but $\C$ may not be \emph{locally minimal}.
\end{Remark}

\begin{restatable}{lemma}{LemmaCharTwo}
\label{lemma_lm_char2}
Let $\C = \{C_1, C_2, ..., C_k\}$ be a clique cover of a subset of edges of $G=(V,E)$. If $\C$ is \emph{locally minimal}, then $C_i \not\subseteq C_j, \forall i,j \in [1,k], i \ne j$ (no clique in $\C$ is contained in another clique).
\end{restatable}
\begin{proof}
We prove the contrapositive, i.e., if for some $i,j \in [1,k], i \ne j$, $C_i \subseteq C_j$, then $\C$ is not locally minimal.

WLOG assume $C_i$ is being created first, and consider the first uncovered edge $\{x,y\}$ for which $C_j$ is being created. Since $C_i \subseteq C_j$, at any stage of the construction, for any $ z \in C_i$, three cases can arise: (1) $x \in C_i$, $\{y,z\} \in E$, (2) $y \in C_i$, $\{x,z\} \in E$, (3) $x \not\in C_i$, $y \not\in C_i$ and $\{x,z\} \in E$, $\{y, z\} \in E$. Clearly, for any of the three cases, rule (i) of Definition \ref{local_min} applies for $\{x,y\}$, but, rule (ii) of Definition \ref{local_min} is being applied. Therefore, $\C$ is not \emph{locally minimal}.
\end{proof}

\begin{Remark}
The converse of characterization in Lemma \ref{lemma_lm_char2} does not hold; i.e., it is possible that no clique in $\C$ is contained in another clique, but $\C$ may not be \emph{locally minimal}.
\end{Remark}

\begin{restatable}{lemma}{LemmaCharThree}
\label{lemma_lm_char3}
Let $\C = \{C_1, C_2, ..., C_k\}$ be a clique cover of a subset of edges of $G=(V,E)$. If $\C$ is \emph{locally minimal}, then $\forall \{C_i, C_j\} \in \C, i \ne j, \exists  x \in C_i, y \in C_j$, $x \neq y$ such that $\{x,y\} \notin E$ (every pair of cliques in $\C$ corresponds to a pair of vertices not connected by an edge in $G$).
\end{restatable}
\begin{proof}
We prove the contrapositive, i.e., if for some $\{C_i,C_j\} \in \C$ with $i \ne j$, this is the case that $\forall x \in C_i$ and $\forall y \in C_j$, $x \ne y$, $\{x,y\} \in E$, then $\C$ is not locally minimal.

WLOG assume $\C_i$ is being created first, and consider the first uncovered edge $\{y, z\}$ for which $\C_j$ is being created. Clearly, $\forall x \in \C_i$, $\{x,z\} \in E$ and $\{x,y\} \in E$, and rule (i) of Definition \ref{local_min} applies for $\{y,z\}$, instead rule (ii) of Definition \ref{local_min} is being applied. Therefore, $\C$ is not \emph{locally minimal}.
\end{proof}

\begin{Remark}
\label{rem_3}
\emph{Clique cover} satisfying Lemma \ref{lemma_lm_char2} does not imply it satisfies Lemma \ref{lemma_lm_char3}. But, \emph{clique cover} satisfies Lemma \ref{lemma_lm_char3} implies no clique is contained in another clique, i.e., Lemma \ref{lemma_lm_char3} is a stronger characterization than Lemma \ref{lemma_lm_char2}.
\end{Remark}

Note that it is the combination of the characterizations in Lemma \ref{lemma_lm_char1}, Lemma \ref{lemma_lm_char2}, and Lemma \ref{lemma_lm_char3} that makes a \emph{locally minimal clique cover} desirable than each standalone characterization. For example, it is easy to satisfy the characterizations in Lemma \ref{lemma_lm_char1} and Lemma \ref{lemma_lm_char2} for a complete graph by covering every edge with a clique containing only the edge. But, characterization in Lemma \ref{lemma_lm_char3} ensures that a single clique covers the complete graph. We can also glean more information on \emph{locally minimal clique cover} using these characterizations further. The following is a characterization of possible redundancy of cliques in a \emph{locally minimal clique cover} in respect of a \emph{minimal clique cover}.

\begin{restatable}{lemma}{LemmaCharFour}
\label{lemma_lm_char4}
Let $\C = \{C_1, C_2, ..., C_k\}$ be a \emph{locally minimal clique cover} of $G=(V,E)$, and for any clique $C_l \in \C$, let $\C^{'} \subseteq \C$ such that $C_l \not\in \C^{'}$. Let $A$ denote the set of edges contained in $C_l$, and let $B$ denote the union of the edges contained in the cliques of $\C^{'}$. If $A \subseteq B$, then $|\C^{'}| \geq 3$.
\end{restatable}
\begin{proof}
We prove the contrapositive, i.e., if $|\C^{'}| < 3$, then $A \not\subseteq B$.

If $|\C^{'} | = 1$, then by Lemma \ref{lemma_lm_char2}, $C_l$ cannot be contained in another clique in $\C$. Hence let $\C^{'} = \{C_i, C_j\}$.

If $|C_l| = 2$, then the single edge in $C_l$ cannot be contained in any other clique, otherwise it would violate the \emph{locally minimal} property of $\C$. Therefore, $|C_l| > 2$.

Let $D = C_l \cap C_i$ and $F = C_l \cap C_j$. If $D=F$ then it must be the case that $|C_l| > |D|$ and $|C_l| > |F|$; otherwise $\C_l \subseteq C_i$ and $\C_l \subseteq C_j$, i.e., $\C$ is not \emph{locally minimal}. Therefore, we can assume $D \ne F$.

$D \not\subseteq F$ and $F \not\subseteq D$; otherwise, $C_l \subseteq C_i$ or $C_l \subseteq C_j$, and $\C$ is not \emph{locally minimal}. Since $D \not\subseteq F$ and $F \not\subseteq D$, it must be the case that for some $x \in D$, $x \not\in F$, and also for some $y \in F$, $y \not\in D$. Edge $\{x,y\} \in A$, but neither $C_i$ not $C_j$ contains the edge $\{x,y\}$.
\end{proof}

Lemma \ref{lemma_lm_char4} implies that in order to remove any redundant clique from a \emph{locally minimal clique cover}, we would need to consider at least three cliques at a time. For a clique containing many vertices, we may need to consider a large number of cliques at a time to discern whether a clique is redundant in a \emph{locally minimal clique cover}.

We obtain bounds for total number of (repeatable) vertices and total number of (repeatable) edges in a \emph{locally minimal clique cover} as follows.

\begin{restatable}{lemma}{LemmaLMCCvertices}
\label{lemma_gsum1}
Let $\C = \{C_1, C_2, ..., C_k\}$ be a clique cover of a subset of edges of $G=(V,E)$. If $\C$ is locally minimal, then $\sum_{C_l \in \C} |C_l| = O(m)$.
\end{restatable}
\begin{proof}
From Lemma \ref{lemma_lm_char1}, any vertex $x$ can appear at most $|N(x)|$ times in the cover $\C$.
Therefore $\sum_{C_l \in \C} |C_l| = \sum_{x \in V} |\{C_l | x \in C_l\}| \leq \sum_{x \in V} |N(x)| = O(m)$.
\end{proof}

\begin{restatable}{lemma}{LemmaLMCCedges}
\label{lemma_gsum2}
Let $\C = \{C_1, C_2, ..., C_k\}$ be a clique cover of a subset of edges of $G=(V,E)$. If $\C$ is locally minimal, then $\sum_{\{x,y\} \in E} |\{C_l | \{x,y\} \in C_l, C_l \in \C\}| = O(m\Delta)$.
\end{restatable}
\begin{proof}
Consider any edge $\{x,y\} \in E$. Since we know from Lemma \ref{lemma_lm_char1} that $x \in V$ can appear at most $|N(x)|$ times in $\C$, both $x$ and $y$ can appear at most $\min\{|N(x)|, |N(y)|\}$ times in $\C$.
Summing over all the edges we get $\sum_{\{x, y\} \in E} \min{\{|N(x)|, |N(y)|\}} \leq \sum_{\{x, y\} \in E} \Delta = O(m\Delta)$. It follows that

$$\sum_{\{x,y\} \in E} |\{C_l | \{x,y\} \in C_l, C_l \in \C\}| $$

$\leq \sum_{\{x, y\} \in E} \min{\{|N(x)|, |N(y)|\}} = O(m\Delta)$
\end{proof}

\begin{Definition}[Intersection Graph Basis]
Let $X = \{1, 2,..., k\}$, and consider a family of $n$ sets $\mathcal{F} = \{F_1, F_2, ..., F_n\}$, with $F_i \subseteq X$. The \emph{intersection graph} of $\mathcal{F}$ has vertices $F_i,  i \in [1, n]$, and an edge joins $F_i$ and $F_j$ when $F_i \cap F_j \ne \phi$, with $i \ne j$. Every graph $G = (V, E)$ arises as the intersection graph of some family of sets; i.e., there is an assignment of a set $F_i$ to each vertex $i$ of $V$ such that for $\{i, j\} \in V$, and $i \ne j$, the following holds: $\{i,j\} \in E \Leftrightarrow F_i \cap F_j \ne \phi$. Computing \emph{(edge) clique cover} of $G$ is equivalent to computing a set $X$ such that $G$ is the intersection graph of a family of subsets of $X$ \cite{erdos1966representation}. Such a set $X$ with minimum cardinality is called an intersection graph basis of $G$ \cite{garey1979computers}. Every element $l \in X$ corresponds to a clique such that $\{F_i |l \in F_i\}$ is the vertex set of the clique.
\end{Definition}

Next, we introduce concepts related to family of sets representing vertices of an \emph{intersection graph}.

\begin{Definition}[Candidate Clique]
\label{cand_cliq}
Let $\C = \{C_1, C_2, ..., C_k\}$ be a \emph{clique cover} of a subset of edges of $G=(V,E)$. $C_l \in \C$ is a \emph{candidate clique} of $x \in V$, if either (i) $x \in C_l$, or (ii) $x \not\in C_l$ and $C_l \subseteq N(x)$.
\end{Definition}

Consider the example graph in Figure \ref{fig:example_graph}, and a clique cover $\C=\{\{1,3,5\}, \{4, 5, 6\}\}$ of a subset of edges of the graph. The clique $\{1,3,5\}$ is a \emph{candidate clique} for all the vertices contained in the clique. The clique $\{4,5,6\}$ is a \emph{candidate clique} for vertex 3, since the vertices in the clique are contained among the neighbors of vertex 3. 

\begin{Definition}[Candidate Clique Set]
Let $\C = \{C_1, C_2, ..., C_k\}$ be a \emph{clique cover} of a subset of edges of $G=(V,E)$. A \emph{candidate clique set} $S_x$ for a vertex $x$ is the set of (indices of) \emph{candidate cliques} of $x$ in $\C$, i.e., $S_x = \{l | x \in C_l \textit{ or } (x \not\in C_l \textit{ and } C_l \subseteq N(x)), 1 \leq l \leq k\}$.
\end{Definition}

\begin{Remark}
For a family of sets $\mathcal{F}$ such that $G=(V,E)$ is the intersection graph of $\mathcal{F}$, it is the case that $\{x,y\} \in V, x \ne y$ following holds: $F_x \cap F_y \ne \phi \Leftrightarrow  \{x, y\} \in E$. For \emph{candidate clique sets}, $S_x \cap S_y \ne \phi$ does not imply $\{x,y\} \in E$.
\end{Remark}

For the example graph of Figure \ref{fig:example_graph}, let $\C = \{C_1, C_2, C_3\}$ be a clique cover of a subset of edges of the graph, where $C_1 = \{1,3,5\}$, $C_2 = \{4,5\}$, and $C_3 = \{2, 4\}$. The \emph{candidate clique set} for vertex 3, $S_3 = \{1, 2, 3\}$, since vertex 3 is contained in $C_1$, and the vertices of both $C_2$ and $C_3$ are contained among the neighbors of vertex 3. Vertex 6 is not contained in any clique of $C$, but $C_2$ is contained among the neighbors of vertex 6, and $S_6 = \{2\}$.

We can characterize the set of cliques that can cover an uncovered edge, precisely through \emph{candidate clique sets} as follows.

\begin{restatable}{lemma}{LemmaCover}
\label{lemma_cover}
Let $\C = \{C_1, C_2, ..., C_k\}$ be a clique cover of a subset of edges of $G=(V,E)$, and let $S_x$ be the \emph{candidate clique set} of a vertex $x$ in $V$. Let $\{x,y\} \in E$ be an uncovered edge of $G$. For $l \in [1,k]$, $l \in S_x \cap S_y$ if and only if $\{x, y\}$ can be covered by a clique contained in $\C$.
\end{restatable}
\begin{proof}
{\bf Only if}: Since the edge $\{x, y\}$ is not covered by the cliques in $\C$, both $x$ and $y$ are not contained in any of the cliques of $ \C$. Then by the definition of \emph{candidate clique set}, one of the following holds for $C_l$:
(i) $x \in C_l$, $y \not\in C_l$, $C_l \subseteq N(y)$.
(ii) $ y \in C_l$, $x \not\in C_l$, $C_l \subseteq N(x)$.
(iii) $x \not\in C_l$, $y \not\in C_l$, $C_l \subseteq N(x)$, $C_l \subseteq N(y)$.

No matter which of the three cases holds, both vertices $x$ and $y$ can be included in $C_l$, and thus $\{x,y\}$ can be covered by $C_l$.

{\bf If}: Since the edge $\{x,y\}$ can be covered by a clique contained in $\C$, $\exists C_l \in \C$ in which $x$ or $y$ or both are not contained, but can be added to. Since $\{x,y\}$ is an uncovered edge, if $x$ is contained in $C_l$ then it must be the case $ y \not\in C_l$ and $C_l \subseteq N(y)$. It follows that if $x$ is contained in $C_l$ then $l \in S_x$ and $l \in S_y$. Similarly, $l \in S_x$ and $l \in S_y$ for the cases when $y \in C_l$ or both $x$ and $y$ are not contained in $C_l$. Therefore, $l \in S_x \cap S_y$.
\end{proof}

\begin{Remark}
\label{rem_1}
For an uncovered edge $\{x,y\}$, consider the choices of covering the edge with a clique $C_l$ contained in $\C$ (clique cover at hand) as follows.
\newline
(i) $x \in C_l$, $y \not\in C_l$, $C_l \subseteq N(y)$.
\newline
(ii) $ y \in C_l$, $x \not\in C_l$, $C_l \subseteq N(x)$.
\newline
(iii) $x \not\in C_l$, $y \not\in C_l$, $C_l \subseteq N(x)$, $C_l \subseteq N(y)$.
\newline
Assume, we choose to utilize only choice (i) to cover $\{x,y\}$. Then, we are essentially restricted to order the vertices in $V$ in some order such that vertex $x$ is placed before vertex $y$ in the order. Conversely, ordering of vertices (or edges) may restrict our choices to cover $\{x,y\}$. \emph{Candidate clique} concept abstracts away any restriction of choices that may be imposed by any ordering of vertices (or edges), and presents us with all potential choices at hand.
\end{Remark}

We obtain following result that helps us to bound several time and space related quantities.

\begin{restatable}{lemma}{LemmaCSum}
\label{lemma_csum}
Let $\mathcal{C} = \{\mathcal{C}_1, \mathcal{C}_2, ..., \mathcal{C}_k\}$ be a clique cover of a subset of edges of $G=(V,E)$. If $\C$ is \emph{locally minimal}, then $\sum_{x \in V} |S_x| = O(m + |\C|\Delta)$.
\end{restatable}
\begin{proof}
Using the two conditions that a \emph{candidate clique} of a vertex could satisfy (Definition \ref{cand_cliq}), we can express the sum as follows.
\begin{equation}
\label{eq:01}
  \sum_{x \in V} |S_x| = \sum_{x\in V} |\{l | x \in C_l\}| + \sum_{x \in V} |\{l | x \not\in C_l, C_l \subseteq N(x)\}|  
\end{equation}

We can bound the first sum on the right side using Lemma \ref{lemma_gsum1} to obtain 
$$\sum_{x\in V} |\{l | x \in C_l\}| = \sum_{C_l \in \C} |C_l| = O(m)$$
We obtain a bound for the second term in the right-hand-side of Equation \ref{eq:01} as follows.

Since $\C$ is \emph{locally minimal}, we must have created every new clique $C_{k+1}$ using rule (ii) of Definition \ref{local_min}. Immediately after covering edge $\{x,y\}$, it must be the case that $k+1 \in S_z, \forall z \in (N(x) \cap N(y)) \cup \{x, y\}$.

On the other hand, whenever applicable we used rule (i) of Definition \ref{local_min} to expand the clique cover by an uncovered edge $\{x,y\}$. From Lemma \ref{lemma_cover}, this corresponds to having $l \in S_x \cap S_y$. But $C_l$ is an existing clique, and all \emph{candidate clique sets} that include $l$ have already been accounted for when we created $C_l$ using rule (ii) of Definition \ref{local_min}.

Therefore, any vertex $z$ such that $z \not\in C_l$, but $C_l \subseteq N(z)$, must be a common neighbor of $x$ and $y$, such that the edge $\{x,y\}$ required the creation of $C_l$ using rule (ii) of Definition \ref{local_min}. And $|\{z| z \not\in C_l, C_l \subseteq N(z)\}| \leq |N(x) \cap N(y)| = O(\Delta)$.

Summing over all cliques we get $$\sum_{C_l \in \C} |\{x |x \not\in C_l, C_l \subseteq N(x)\} | = O(|\C|\Delta)$$ The sum on the left hand side of preceding equality is another way to express the sum for which we needed a bound. To see this consider a bipartite graph with vertices in $V$ on one side, and the cliques in $\C$ on the other sides. In this bipartite graph, an edge connects a vertex to a clique if condition (ii) of Definition \ref{cand_cliq} holds for the clique. Sum of degrees from one side of this bipartite graph is equal to the sum of the degrees from the other side. Hence we have the following equality.
$$\sum_{x \in V} |\{l | x \not\in C_l, C_l \subseteq N(x) \}|$$
$$ = \sum_{C_l \in \C} |\{x | x \not\in C_l, C_l \subseteq N(x)\} | $$

It follows that 

$\sum_{x \in V} |\{l | x \not\in C_l, C_l \subseteq N(x) \}|= O(|\C|\Delta)$
\end{proof}

\begin{property}
\label{prop_1}
Consider a graph $G$ whose connected components are triangle-free. Every edge of $G$ has to be covered with a distinct clique. Therefore, every clique would contribute to $S_x$ values of exactly 2 vertices, the end vertices of the edge contained in the clique. And, $\sum_{x \in V} S_x = \sum_{x \in V} |N(x)| = 2m = \Theta(m)$. A similar bound is true for dense graphs, e.g., for a complete graph $m = n^2$, and $\sum_{x\in V} |S_x| = n = O(m)$.
\end{property}

Using connection to \emph{intersection graph basis}, we obtain following connection between \emph{minimum clique cover} and \emph{candidate clique sets}.

\begin{restatable}{lemma}{LemmaIGB}
\label{lemma_igb}
Let $\C = \{C_1, C_2, ..., C_k\}$ be a clique cover of $G=(V,E)$, and let $S_x$ be the \emph{candidate clique set} of a vertex $x$ in $V$. If $|\cup_{x \in V} S_x|$ is minimum over the family of \emph{candidate clique sets} resulting from the clique covers of $G$, then $\C$ is a \emph{minimum clique cover} of $G$.
\end{restatable}
\begin{proof}
Let $F_x = S_x \backslash \{l | x \not\in C_l\}$, $Y = \cup_{x \in V} F_x$, and $Z = \cup_{x\in V} S_x$.

Note that $F_x = \{l | x\in C_l\}$, and by our assumption $|Z|$ is minimum. Therefore, each of the clique indices in $Z$ is contained in $Y$, i.e., $|Y|= |Z|$.

$Y$ is an intersection graph basis of $G$, and $\{F_1, F_2, ..., F_n\}$ is the family of subsets of $Y$ where $F_i$ corresponds to vertex $i \in V$. To see this assume $|Y|$ is not minimum, i.e., $\exists l \in Y$, such that $C_l$ can be removed from $\C$. Which contradicts our assumption that $|Z|$ is minimum.

Since $Y$ is an intersection graph basis of $G$, $\C$ is a \emph{minimum clique cover} of $G$ [\cite{erdos1966representation}, \cite{roberts1985applications}].
\end{proof}
\section{Greedy Framework}
\label{sec:greedy}

We describe a greedy framework that computes \emph{locally minimal clique cover}. We show that \emph{locally minimal clique cover} computed by our greedy framework can be efficiently turned into \emph{minimal clique cover}. We start with a basic greedy framework, and subsequently modify the basic framework to derive an improved greedy framework with better time bound. We conclude the section showing modifications of the improved greedy framework that make the time bound attuned to degeneracy of graph.

\subsection{Basic Greedy Framework.}

\begin{algorithm}
\caption{\textit{Basic Greedy Framework}}
\label{alg:dineB}
\textbf{Input:} A graph $G=(V,E)$\\
\textbf{Output:} A \emph{locally minimal clique cover} $\C$ of $G$
\begin{algorithmic}[1]
\STATE{Select an ordering of edges in $E$, and process edges in that order in steps 3-13}
\STATE{$k \leftarrow 0$}
\WHILE{$\{x,y\} \in E$ uncovered}
\STATE{$S_x \leftarrow \{l | C_l$ is \emph{candidate clique} of $x$, $1\leq l \leq k\}$}
\STATE{$S_y \leftarrow \{l | C_l$ is \emph{candidate clique} of $y$, $1 \leq l \leq k \}$}
\IF{$S_x \cap S_y \ne \phi$}
\STATE{Select $l \in S_x \cap S_y$}
\STATE{$C_l \leftarrow C_l \cup \{x,y\}$}
\ELSE
\STATE{$k \leftarrow k + 1$}
\STATE{$C_k \leftarrow \{x,y\}$}
\ENDIF
\ENDWHILE
\RETURN $\C$ 
\end{algorithmic}
\end{algorithm}

Our basic greedy framework is described in Algorithm \ref{alg:dineB}. Algorithm \ref{alg:dineB} strictly follows the two rules defined for \emph{locally minimal clique cover}. To apply rule (i) of Definition \ref{local_min}, it uses characterization described in Lemma \ref{lemma_cover}. At every iteration, it computes set of \emph{candidate cliques} for two end vertices of an uncovered edge $\{x,y\}$ (step 4 and 5). If there is an index of a clique $C_l$ contained in the \emph{candidate clique sets} of $x$ and $y$, it covers the edge using $C_l$ (step 7 and 8). Otherwise, it creates a new clique containing only $\{x,y\}$ (step 10 and 11). After the execution of steps 1 - 13, we have following result.

\begin{restatable}{lemma}{LemmaGreedyOne}
\label{lemma_greedy1}
For a graph $G=(V,E)$, Algorithm \ref{alg:dineB} computes a \emph{locally minimal clique cove}r $\C$ of $G$.
\end{restatable}
\begin{proof}
After every iteration of steps 4-12, Algorithm \ref{alg:dineB} maintains the invariant that the \emph{clique cover} computed for the subset of edges covered after the iteration is \emph{locally minimal}. This is straightforward to see as steps 8 and 11 are the only steps where we modified the \emph{clique cover}.

The invariant holds for either of the steps: step 8 directly follows from rule (i) of Definition \ref{local_min}, and Lemma \ref{lemma_cover}, and, step 11 is a straightforward use of rule (ii) of Definition \ref{local_min}. Therefore, after the execution of steps 1-13 we have a \emph{locally minimal clique cover}.
\end{proof}

For the data structures, we assume that the graph is represented by its adjacency lists, and that each list is stored in a hash table. With universal hashing or simple uniform hashing, an element can be accessed, or inserted, or deleted in $O(1)$ average time per operation. We can also convert every list to an ordered list taking $O(m\log\Delta)$ time in total, and then insertion, or deletion, or access of any element would take $O(\log \Delta)$ time. For the \emph{candidate clique sets}, we assume similar data structure as adjacency lists, where each list maintains \emph{candidate clique set} of a vertex.

For step 7 in Algorithm \ref{alg:dineB}, one can do substantial computation per selection without affecting our obtained time bounds; but, for all our analysis, we stick to simple choices such as selecting a smallest clique, or a largest clique, or clique created earliest, or a random clique. For the next analysis we assume an edge ordering as follows. Sort the vertices from lowest to highest degree and then select vertices in that order. For each selected vertex, cover all the edges incident on the vertex. With the choices of edge ordering and clique selection, we obtain following result.

\begin{restatable}{lemma}{LemmaGreedyOneb}
\label{lemma_greedy1b}
For a graph $G=(V,E)$, Algorithm \ref{alg:dineB} computes a \emph{locally minimal clique cover} $\C$ of $G$, taking $O(m^2)$ time and $O(m)$ space.
\end{restatable}
\begin{proof}
Lemma \ref{lemma_greedy1} shows the proof of \emph{locally minimal clique cover}. Now, we show the analysis for time and space bounds.

For step 1, our choice of edge ordering takes $O(n\log n)$ time. Isolated vertices would not affect execution of Algorithm $\ref{alg:dineB}$. Therefore, for rest of the steps we can assume $n=O(m)$, and time bound for this step would be dominated by the time bounds of the other steps shown next.

For steps 4 and 5, for each clique $C_l$ we need to check whether $C_l$ is contained in the neighborhood of $x$ and $y$, or whether $x$ and $y$ are contained in $C_l$. For each such $C_l$, it takes $O(min\{|C_l|, |N(x)|, |N(y)|\})$ time. Combining the fact that $\sum_{l=1}^k |C_l| = O(m)$ from Lemma \ref{lemma_gsum1}, we get time bounds of steps 4 and 5 for the entire execution of the algorithm as follows. $O(\sum_{\{x,y\} \in E} \sum_{l=1}^k min\{|C_l|, |N(x)|, |N(y)|\}) =O( \sum_{\{x,y\} \in E} \sum_{l=1}^k |C_l|) = O(\sum_{\{x,y\} \in E} m) = O(m^2)$.

For our choices of clique selection, step 7 takes $O(m)$ time, since, $|\C| = O(m)$ and $|S_x| = |S_y| = O(m)$. And, step 8 takes $O(\Delta)$ time: to make sure that any uncovered edge among $x$, $y$ and vertices of $C_l$ is being identified as covered. Steps 10 and 11 take $O(1)$ time.

It is straightforward to see that the algorithm utilizes at most two sets of \emph{candidate clique} sets at steps 4 and 5, and each of them takes $O(m)$ space. 

Combining preceding bounds, we conclude that Algorithm \ref{alg:dineB} does $O(m^2)$ computations, using $O(m)$ space.
\end{proof}

\subsubsection{Computing \emph{Minimal Clique Cover}.}
\label{sec_pp}

In light of the characterization in Lemma \ref{lemma_lm_char4}, we need efficient ways for finding cliques that may be redundant in a \emph{locally minimal clique cover}. We can obtain a \emph{minimal clique cover} from an arbitrary clique cover in different ways. For example, we can compute all removable assignments (vertices or edges that can be removed) in a clique cover similar to \cite{ennis2012assignment}, and remove assignments repeatedly until no such assignments are removable. Here, we are going to exploit the characterization we obtained for \emph{locally minimal clique cover} in Lemma \ref{lemma_gsum2}. We describe fast realization of a post-processing step proposed by \cite{kou1978covering} to get \emph{minimal clique cover}, and analyze the step for \emph{locally minimal clique cover}.

\begin{algorithm}
\caption{Post-Process}
\label{alg:pp}
\textbf{Input:} A graph $G=(V,E)$, and a \emph{clique cover} $\C$ of $G$\\
\textbf{Output:} A \emph{minimal clique cover} $\C^{'}$ of $G$
\begin{algorithmic}[1]
\STATE{$\forall \{x,y\} \in E$ compute $AppCount(\{x,y\})$ denoting number of times $\{x,y\}$ appears in $\C$}
\STATE{$\C^{'} \leftarrow \phi$}

\FOR{\textbf{each} $C_l \in \C$}
\IF{$AppCount(\{x,y\}) > 1$, $\forall \{x,y\} \in C_l, x \ne y$}
\STATE{$AppCount(\{x,y\}) \leftarrow AppCount(\{x,y\}) - 1$, $\forall \{x,y\} \in C_l, x \ne y$ }
\ELSE
\STATE{$\C^{'} \leftarrow \C^{'} \cup \{C_l\}$}
\ENDIF
\ENDFOR
\RETURN $\C^{'}$ 
\end{algorithmic}
\end{algorithm}

Algorithm \ref{alg:pp} shows our realization of post-processing step proposed by \cite{kou1978covering}. At step 1, it enumerates all edges contained in the cliques of $\C$, and for every edge of graph, starts keeping track of number of appearances of the edge in $\C$. If for a clique $C_l$, all of its edges appear more than once in $\C$ (step 4), then we discard $C_l$ and decrease appearance counts for all the edges contained in $C_l$ (step 5). Otherwise, we add $C_l$ to the \emph{minimal clique cover} (step 7). Following is straightforward to show.

\begin{restatable}{lemma}{LemmaPP}
\label{lemma_pp}
A \emph{minimal clique cover} with all the characterizations of \emph{locally minimal clique cover} can be obtained from a \emph{locally minimal clique cover} in $O(m\Delta)$ time, using $O(m)$ space.
\end{restatable}
\begin{proof}
\emph{Minimal Clique Cover}. When step 4 is true, $\C_l$ can be discarded, and we do not include $C_l$ in minimal clique cover $\C^{'}$. In step 5, accounting for the fact that $C_l$ is discarded, we affect that by decreasing count of $C_l$ for the edges contained in $C_l$. On the other hand, if for any $C_l$, step 4 is false, then we know $C_l$ contains at least one edge that is not covered by remaining cliques in $C$, and $C_l$ must be included in the \emph{minimal clique cover}. We make sure this in step 7.

Time and Space. From Lemma \ref{lemma_gsum2}, we know total number of appearances of edges in a locally minimal clique cover is $O(m\Delta)$. Since we only enumerate every edge contained in $\C$ once in step 1, once in step 4, and at most once in step 5, entire post-processing step takes $O(m\Delta)$ time. We only use $O(m)$ counters initialized at step 1, therefore, space bound is $O(m)$.
\end{proof}

\subsection{Improved Greedy Framework.}

Our basic greedy framework sets the stage for computing \emph{locally minimal clique cover}, thereby \emph{minimal clique cover}, in simplest possible terms. Now, we describe modifications of the basic framework that saves us from repeated computations of \emph{candidate clique sets} in step 4 and 5 of Algorithm \ref{alg:dineB}. To this end, we outline two parts incremental computations for \emph{candidate clique sets} as follows.

\begin{enumerate}
    \item Whenever a new clique  $C_k = \{x,y\}$ is created, we add $k$ to the \emph{candidate clique sets} of common neighbors of $x$ and $y$; and also to the \emph{candidate clique sets} of $x$ and $y$.
    \item Whenever an existing clique $C_l$ is extended by an uncovered edge $\{x,y\}$, we remove $l$ from the \emph{candidate clique set} of a vertex $z$, such that both vertices $x$ and $y$ are not in $N(z)$. And we do this for all such vertices, except for the cases $z = x$ or $z = y$.
\end{enumerate}

For the first of part of computation since we just have created $C_k = \{x,y\}$, for every common neighbor of $x$ and $y$, say $z$, $C_k = \{x, y\}$ is contained in $N(z)$. Therefore, $k$ must be added to $S_z$.

For the second of part of computation since we just have included $\{x,y\}$ to $C_l$, every $z \in V$ such that $l$ is contained in $S_z$, must contain both $x$ and $y$ in $N(z)$. Otherwise, we must remove $l$ from $S_z$, and restore validity of all \emph{candidate clique sets}.

For the second part, to find $z$ so that we can remove $l$ from $S_z$, we select a vertex from $C_l$, say $w$, arbitrarily. Because, every such $z$ is a common neighbor of vertices contained in $C_l$, processing only the neighbors of $w$, $N(w)$, would suffice.

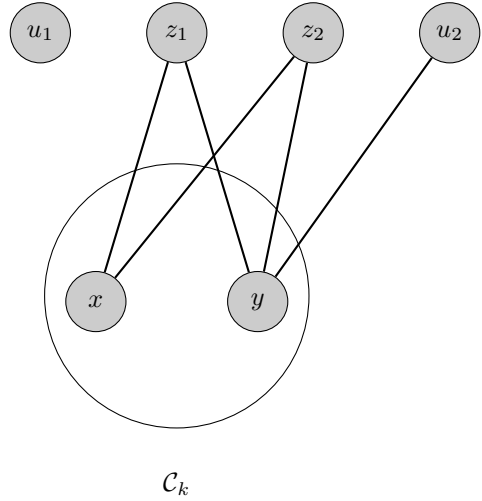
\begin{figure}[!t]
    \centering
     \tikzset{main node/.style={circle,fill=black!20,draw,minimum size=0.8cm,inner sep=0pt},
            }
 \begin{tikzpicture}
    \node[main node] (1) {$z_1$};
    \node[main node] (2) [right = 1.0 cm of 1] {$z_2$};
    
    \node[main node] (3) [below left = 3.0cm and 0.5cm of 1] {$x$};
    \node[main node] (4) [below right = 3.0cm and 0.5cm of 1] {$y$};
    
    \node[main node] (5) [right = 1.0 cm of 2] {$u_2$};
    
    \node[main node] (6) [left = 1.0 cm of 1] {$u_1$};

    \draw (0,-3.5) circle (50pt);
    
    \node[] at (0,-6.0) {$\C_k$};
    
    \path[draw,thick]
    (1) edge node {} (3)
    (1) edge node {} (4)
    (2) edge node {} (3)
    (2) edge node {} (4)
    (5) edge node {} (4)
    ;
    
\end{tikzpicture}
    \caption{Except $u_1$ and $u_2$, all other vertices would have $C_k$ as \emph{candidate clique} after the creation of $C_k$}
    \label{fig:dine1}
\end{figure}

\begin{figure}[!t]
    \centering
     \tikzset{main node/.style={circle,fill=black!20,draw,minimum size=0.8cm,inner sep=0pt},
            }

\tikzset{clique node/.style={circle,fill=black!0,draw,minimum size=4.0cm,inner sep=0pt},
            }
 \begin{tikzpicture}
    
    \node[clique node] (8) {};
    
    \node[main node] (1) [above right = 1.5cm and 0.2cm of 8] {$z_1$};
    \node[main node] (2) [right = 1.0 cm of 1] {$z_2$};
    
    \node[main node] (3) [below left = 1.5cm and 1.5cm of 1] {$w_1$};
    \node[main node] (4) [below left= 1.0cm and 0.4cm of 3] {$w_2$};
    \node[main node] (5) [below right = 1.0cm and 0.4cm of 3] {$w_3$};
    
    \node[main node] (7) [below = 1.0 cm of 2] {$y$};
    
    \node[main node] (6) [below = 1.5 cm of 7] {$x$};


    \node[] at (0,-2.5) {$\C_l$};
    
    \path[draw,line width=0.05mm]
    (6) edge node {} (7)
    (1) edge node {} (6)
    (1) edge node {} (7)
    (2) edge node {} (7)
    ;
    
    \path[draw,line width=0.5mm]
    (1) edge node {} (8)
    (2) edge node {} (8)
    (6) edge node {} (8)
    (7) edge node {} (8)
    ;
    
\end{tikzpicture}
    \caption{While including $x$ and $y$ to $C_l$, we would remove $l$ from \emph{candidate clique set} of $z_2$, $S(z_2)$, since $x$ is not in $N(z_2)$. Thick edge between a vertex and the clique $C_l$ denotes the vertices in the clique are contained in the neighborhood of the vertex.}
    \label{fig:dine2}
\end{figure}
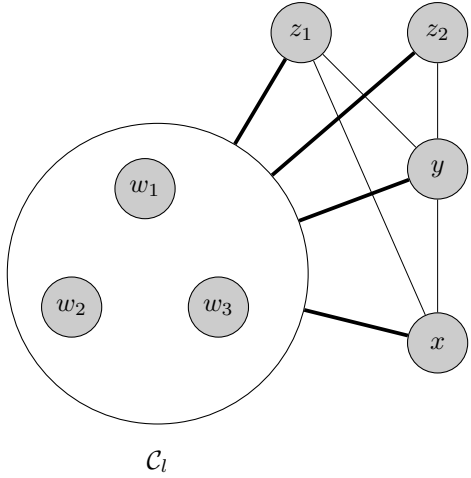

\begin{algorithm}
\caption{\textit{Improved Greedy Framework}}
\label{alg:dine}
\textbf{Input:} A graph $G=(V,E)$\\
\textbf{Output:} A \emph{minimal clique cover} $\C$ of $G$
\begin{algorithmic}[1]
\STATE{Select an ordering of edges in $E$, and process edges in that order in steps 3-13}
\STATE{$k \leftarrow 0$}
\WHILE{$\{x,y\} \in E$ uncovered}
\IF{$S_x \cap S_y \ne \phi$}
\STATE{Select $l \in S_x \cap S_y$}
\STATE{Select a vertex $w \in C_l$ arbitrarily. $\forall z \in N(w) \backslash \{x,y\}$, remove $l$ from $S_z$, if both $x$ and $y$ are not in $N(z)$}
\STATE{$C_l \leftarrow C_l \cup \{x,y\}$}
\ELSE
\STATE{$k \leftarrow k + 1$}
\STATE{$C_k \leftarrow \{x,y\}$}
\STATE{Add $k$ to $S_z$, $\forall z \in (N(x) \cap N(y)) \cup \{x, y\}$}
\ENDIF
\ENDWHILE
\STATE{$\C \leftarrow$ Post-Process($G$, $\{C_1, C_2, ..., C_k\}$)}
\RETURN $\C$ 
\end{algorithmic}
\end{algorithm}

Figure \ref{fig:dine1} shows a schematic for the first part of the computation, and Figure \ref{fig:dine2} shows a schematic for the second part of the computation. Incorporating these two parts incremental computations of \emph{candidate clique sets}, we can discard computation of \emph{candidate clique sets} in steps 4 and 5 of Algorithm \ref{alg:dineB}. Modified algorithm is outlined in Algorithm \ref{alg:dine} with two parts incremental computations added in steps 6 and 11. These modifications of Algorithm \ref{alg:dineB} would compute same \emph{clique cover} as Algorithm \ref{alg:dineB}. Therefore, following is straightforward to show.

\begin{restatable}{lemma}{LemmaGreedyTwo}
\label{lemma_greedy2}
For a graph $G=(V,E)$, Algorithm \ref{alg:dine} computes a \emph{minimal clique cover} $\C$ of $G$.
\end{restatable}
\begin{proof}
All we need to verify that steps 6 and 11 agree with what would have been computed in step 4 and 5 of Algorithm \ref{alg:dineB}.

Step 11 follows from Definition \ref{cand_cliq}, we added index of $C_k$ to exactly the same sets of \emph{candidate clique sets} that would have been computed in step 4 and 5 of Algorithm \ref{alg:dineB}.

For step 6, since we are including $\{x,y\}$ to $C_l$. After this inclusion, according to the definition of \emph{candidate clique}, for any $z\in V$ with $\{z,x\} \not\in E$ or $\{z,y\} \not\in E$, $C_l$ would not have been computed as \emph{candidate clique} in step 4 and 5 of Algorithm \ref{alg:dineB}. And, we make sure to remove $l$ from exactly those $S_z$ at step 6 of Algorithm \ref{alg:dine}.
\end{proof}

With the bound established in Lemma \ref{lemma_csum}, we obtain following characterizations of Algorithm \ref{alg:dine}.

\begin{restatable}{theorem}{ThmCCSG}
\label{thm_ccsg}
For a graph $G=(V,E)$, Algorithm \ref{alg:dine} computes a \emph{minimal clique cover} $\C$ of $G$, taking $O( m\Delta + |\C|\Delta^2)$ time and $O(m + |\C|\Delta)$ space; and for $\sum_{x \in V} |S_x| = O(m)$, it takes $O(m\Delta)$ time and $O(m)$ space.
\end{restatable}
\begin{proof}
Lemma \ref{lemma_greedy2} shows the proof of \emph{minimal clique cover}. Now, we show the analysis for $\sum_{x \in V} |S_x| = O(m)$; the general case follows from replacing the sum with the bound obtained in Lemma \ref{lemma_csum}.

In every iteration, step 4 would take $O(min\{|S_x|, |S_y|\})$ time. All iterations of step 4 would take $O(\sum_{\{x,y\} \in E} min\{|S_x|, |S_y|\}) = O(\sum_{\{x,y\} \in E} |S_x|)$ time. But, any $x\in V$ would contribute at most $|N(x)|$ time in $\sum_{\{x,y\} \in E} |S_x|$, and using the bound $\sum_{x\in V} |S_x| = O(m)$, we get following time bound for all iterations of step 4. $O(\sum_{\{x,y\} \in E} |S_x|) = O(\sum_{x \in V} |S_x| |N(x)|) = O(\Delta \sum_{x \in V} |S_x|) = O(m\Delta)$.

Analysis of step 4 also applies to step 5. And, for our choices of clique selection, step 5 would also take $O(m\Delta)$ time for all iterations.

Step 6 would take $O(1)$ time to select $w$, and $O(\Delta)$ time to remove index of $\C_l$ from all $S_z$ of relevant $z \in N(w)$. Step 7 would also take $O(\Delta)$ time. As shown in Lemma \ref{lemma_pp}, step 14 would take $O(m\Delta)$ time, and $O(m)$ space.

For \emph{candidate clique sets}, we only have used $\sum_{x \in V} |S_x|= O(m)$ space. Combining preceding bounds, we conclude that for $\sum_{x\in V} |S_x| = O(m)$, Algorithm \ref{alg:dine} does $O(m\Delta)$ computations, using $O(m)$ space.
\end{proof}

\subsection{Degeneracy Ordering.}
\label{sec:degen}

We describe modifications of our improved greedy framework that takes degeneracy of input graph into account. For edge ordering in step 1 of Algorithm \ref{alg:dine}, we make use of \emph{degeneracy ordering}, and for graphs of constant degeneracy, show a family of linear time algorithms. To integrate degeneracy in the time bound, we modify the steps outlined in Algorithm \ref{alg:dine} as follows.

\begin{enumerate}
    \item In step 1, compute $u_1, u_2, ..., u_n$ (degeneracy ordering of $V$) and $N_d(u_i), \forall i \in [1,n]$. Cover edges incident from $u_1$ to vertices in $N_d(u_1)$ first, then cover edges incident from $u_2$ to vertices in $N_d(u_2)$, and so on.
    \item In step 6, scan $C_l$, and select a vertex $w$ such that $w$ comes earliest in the degeneracy ordering among the vertices contained in $C_l$. Then instead of $N(w)$ use $N_d(w)$.
    \item In step 11, if $x$ precedes $y$ in the degeneracy ordering of $V$, then use $N_d(x) \cap N(y)$ instead of $N(x) \cap N(y)$. Otherwise, use $N_d(y) \cap N(x)$ instead of $N(x) \cap N(y)$.
    \item Discard step 14.
\end{enumerate}

Note that since every non-empty subgraph must have at least one vertex with degree at most $d$, it must be the case $|C_l| \leq d+1, \forall C_l \in \C$. With the above-stated modifications, the correctness of steps 1-13 of Algorithm \ref{alg:dine} is retained, and we obtain following result.

\begin{restatable}{lemma}{LemmaGreedyThree}
\label{lemma_greedy3}
For a graph $G=(V,E)$, Algorithm \ref{alg:dine} with the above-stated modifications computes a \emph{locally minimal clique cover} $\C$ of $G$.
\end{restatable}
\begin{proof}
Modification of step 1 does not effect validity of \emph{clique cover}. Let for an uncovered edge $\{x,y\}$, $x$ precedes $y$ in the degeneracy order of $V$. With the edge ordering outlined in step 1 of the above-stated modifications, we only need to propagate \emph{candidate clique sets} information to vertices that follow $x$ in the degeneracy ordering of $V$. In step 11, we make sure this by updating \emph{candidate clique sets} of vertices contained in $N_d(x) \cap N(y)$.

In step 6, we make sure we update $S_z$ of all $z$ that are yet to be processed by selecting a vertex $w$ contained in $C_l$ such that among the vertices of $C_l$, $w$ comes earliest in the degeneracy ordering of $V$. Note that this is exactly the vertex $x$ that we added in step 11 when we created the clique with index $l$. Therefore, $N_d(w)$ of step 6 is the same list as of $N_d(x)$ of step 11 when we added $C_l = \{x,y\}$ to the \emph{clique cover}. And, through $N_d(w)$ we get hold of all $z\in V$, whose $S_z$ may contain $C_l$ as \emph{candidate clique}.
\end{proof}

Taking the above-stated modifications into account, we get a tighter bound than the bound of \ref{lemma_csum} as follows.

\begin{restatable}{lemma}{LemmaCSumd}
\label{lemma_csum2}
Let $\mathcal{C} = \{\mathcal{C}_1, \mathcal{C}_2, ..., \mathcal{C}_k\}$ be a \emph{clique cover} of a subset of edges of $G=(V,E)$. If $\C$ is \emph{locally minimal} and computed following the above-stated modifications of Algorithm \ref{alg:dine}, then $\sum_{x \in V} |S_x| = O(m + |\C|d)$.
\end{restatable}
\begin{proof}
Step 1 and 2 of the above-stated modifications do not include any new clique index to any set of $\CCS$. Only step 3 of the above-stated modifications adds clique index to sets of $\CCS$ and we account for the effects as follows. 

For every new clique $C_l$ we created using rule (ii) of Definition \ref{local_min}, there are choices of \emph{candidate clique sets} to which we include clique index $l$. With step 3 of the above-stated modifications, we deviate from the choices of \emph{candidate clique sets} in the proof of Lemma \ref{lemma_csum}. Assume $x$ precedes $y$ in the degeneracy order of $V$. Then with the deviation introduced with step 3 of the above-stated modifications, we include $l$ to the \emph{candidate clique sets} of vertices $z$, $S_z$, such that $z \in (N_d(x) \cap N(y)) \cup \{x, y\}$. Therefore, any vertex $z$ such that $z \not\in C_l$ it must be the case $|\{z | z \not\in C_l, C_l \subseteq N(z)\}| \leq |N_d(x) \cap N(y)| = O(d)$. And we get following change from the proof of Lemma \ref{lemma_csum}.

$\sum_{x \in V} |\{l | x \not\in C_l, C_l \subseteq N(x) \}|= O(|\C|d)$
\end{proof}

We obtain following characterizations of Algorithm \ref{alg:dine} through Lemma \ref{lemma_csum2}.

\begin{restatable}{theorem}{ThmCCSD}
\label{thm_ccsd}
For a graph $G=(V,E)$, Algorithm \ref{alg:dine} with the above-stated modifications computes a \emph{locally minimal clique cover} $\C$ of $G$, taking $O(n+md+|\C|d^2)$ time and $O(m+|\C|d)$ space, and for $\sum_{x \in V} |S_x| = O(m)$, it takes $O(n+md)$ time and $O(m)$ space.
\end{restatable}
\begin{proof}
Lemma \ref{lemma_greedy3} shows the proof of \emph{locally minimal clique cover}. Now, we show the analysis for $\sum_{x \in V} |S_x| = O(m)$; the general case follows from replacing the sum with the bound obtained in Lemma \ref{lemma_csum2}.

Computing degeneracy ordering takes $O(m+n)$ time. Computing $N_d(u_i), \forall i \in [1,n]$ also takes $O(m+n)$ time.

Step 4 would take $O(\sum_{\{x,y\} \in E} min\{|S_x|, |S_y|\})$ $= O(\sum_{\{x,y\} \in E} |S_x|)$ time. Under degeneracy ordering any $x\in V$ would contribute at most $d$ times in $\sum_{\{x,y\} \in E} |S_x|$.

Therefore, step 4 in total would take $O(\sum_{\{x,y\} \in E} |S_x|) = O(\sum_{x\in V} |S_x|d)$ 

$= O(d \sum_{x\in V} |S_x|) = O(md)$ time. Same $O(md)$ time bound applies to step 5, for our choices of clique selection.

Step 6 and 7 would take $O(d)$ time as $|C_l|\leq d+1$, and we are only traversing $N_d(w)$. Step 11 takes $O(|N_d(x) \cap N(y)|) = O(d)$ time.

Combining preceding bounds, we conclude that for $\sum_{x \in V} |S_x| = O(m)$, our modifications of Algorithm \ref{alg:dine} under degeneracy ordering does $O(n+md)$ computations, using $O(m)$ space.
\end{proof}
\section{FPT Algorithms}
\label{sec:fpt}

We present two FPT algorithms for \emph{ECC-d}. Our FPT algorithms are based on two different design paradigms: one algorithm is based on \emph{candidate clique sets} (Section \ref{fpt_1}), and the other algorithm is based on maximal clique enumeration (Section \ref{fpt_2}). To further corroborate versatility of \emph{candidate clique sets} for computing \emph{clique cover} with different objectives, we briefly describe an exact algorithm for the \emph{assignment-minimum clique cover} problem (Section \ref{sec:amcc}).

\subsection{FPT Algorithm Based on \emph{Candidate Clique Sets}.}
\label{fpt_1}
Our basis for computing a \emph{minimum clique cover} of a graph using \emph{candidate clique sets} is laid out in Lemma \ref{lemma_igb}. Here we make sure that the computation takes degeneracy of graph into account. Algorithm \ref{alg:iFPT} solves the parameterized decision problem \emph{ECC-d}: if there is a clique cover of $G$ with no more than $k$ cliques, then Algorithm \ref{alg:iFPT} returns it; otherwise, it returns $\phi$. By iteratively increasing value of $k$, starting from $1$, and calling Algorithm \ref{alg:iFPT} with each value of $k$, we can compute a \emph{minimum clique cover}.

In step 4 of Algorithm \ref{alg:iFPT}, we select an uncovered edge that helps us to bound number of branching required at each node of the search tree. For an uncovered edge selected in step 4, we consider each of the \emph{candidate cliques} that can cover the edge (steps 5-13), and branch into each of such choices (step 8). \emph{Prepare} subroutine (step 6) prepares \emph{candidate clique sets} for inclusion of $\{x,y\}$ in $C_l$, and saves containment status of $x$ and $y$ in $C_l$. If a branch fails at step 8, \emph{Restore} subroutine (step 12) undoes the changes made in steps 6-7. If none of the branches in steps 5-13 succeeds, and parameter $k$ permits, then we try one final branch at steps 14-21. In this branch, we extend the \emph{clique cover} with a new clique containing only $\{x,y\}$ (step 15). If the final branch fails, we restore the \emph{clique cover} (step 20), and return to parent of this node (step 22).

\begin{algorithm}
\caption{\textit{CandidateCliqueFPT}}
\label{alg:iFPT}
\textbf{Input:} A graph $G=(V,E)$, \emph{Clique cover} $\C$ (of a subset of edges in $E$), \emph{Candidate clique sets} $\CCS$, an integer $k$\\
\textbf{Output:} If exists, a \emph{clique cover} $\C$ of $G$ with $|\C| \leq k$; otherwise $\phi$
\begin{algorithmic}[1]
\IF{$\C$ covers $G$}
\RETURN $\C$
\ENDIF
\STATE{Select an uncovered edge $\{x,y\}$ such that $y \in N_d(x)$, and in the degeneracy order of $V$, $x$ follows any other vertex with such a neighbour $y$}
\FOR{\textbf{each} $l \in S_x \cap S_y$}
\STATE{$Prepare(G,\C,\CCS,l,\{x,y\}, W, x_l, y_l)$}
\STATE{$C_l \leftarrow C_l \cup \{x,y\}$}
\STATE{$Q \leftarrow CandidateCliqueFPT(G, \C, \CCS, k)$}
\IF{$Q \ne \phi$}
\RETURN $Q$
\ENDIF
\STATE{$Restore(G,\C,\CCS,l,\{x,y\}, W, x_l, y_l)$}
\ENDFOR
\IF{$k > 0$}
\STATE{$AddNewClique(G,\C,\CCS,\{x,y\})$}
\STATE{$Q \leftarrow CandidateCliqueFPT(G, \C, \CCS, k-1)$}
\IF{$Q \ne \phi$}
\RETURN $Q$
\ENDIF
\STATE{$RemoveNewClique(G,\C,\CCS)$}
\ENDIF
\RETURN $\phi$
\end{algorithmic}
\end{algorithm}

\begin{algorithm}[H]
\caption{\textit{AddNewClique}}
\label{alg:newCl}
\textbf{Input:} A graph $G=(V,E)$, \emph{Clique cover} $\C$ (of a subset of edges in $E$), \emph{Candidate clique sets} $\CCS$, an edge $\{x, y\}$\\
\textbf{Output:} A new clique added to $\C$, and updated \emph{candidate clique sets}
\begin{algorithmic}[1]
\STATE{$p \leftarrow |\C| + 1$}
\STATE{$C_p \leftarrow \{x,y\}$}
\STATE{$R_p \leftarrow (N(x) \cap N(y)) \cup \{x,y\}$}
\FOR{\textbf{each} $z \in R_p$}
\STATE{$S_z \leftarrow S_z \cup \{p\}$}
\ENDFOR
\RETURN $\C, \CCS$
\end{algorithmic}
\end{algorithm}

We could have made copies of $\C$ and $\CCS$ in steps 6 and 15 and could have branched modifying the copies. The algorithm would still be FPT, but we would incur $O(m+k\Delta)$ space and time cost per node in the search tree. Instead, we outline computations in the subroutines, which help us to keep space use down to $O(m+k\Delta)$ in total, and to bound time spent per node to $O(\Delta)$. Next, we expand on the subroutines.

For each clique $C_l$, we use a data structure $R_l$, which keeps track of vertices that have $C_l$ as \emph{candidate cliques}, i.e., $R_l = \{z \in V| l \in S_z\}$. \emph{AddNewClique} subroutine (Algorithm \ref{alg:newCl}) initializes this data structure, for every new clique it adds to the \emph{clique cover}. Also, at each node, \emph{Prepare} and \emph{Restore} subroutines share some data structure, namely $W$, $x_l$ (indicator of whether $C_l$ contains $x$), $y_l$ (indicator of whether $C_l$ contains $y$), to facilitate the modification of $\C$ and $\CCS$ and their restoration.

\begin{algorithm}[H]
\caption{\textit{Prepare}}
\label{alg:prepare}
\textbf{Input:} A graph $G=(V,E)$, \emph{Clique cover} $\C$ (of a subset of edges in $E$), \emph{Candidate clique sets} $\CCS$, clique index $l$, an edge $\{x, y\}$, $W$, $x_l$, $y_l$\\
\textbf{Output:} \emph{Candidate clique sets} prepared for covering $\{x,y\}$ with $C_l$, $W$, $x_l$, $y_l$
\begin{algorithmic}[1]
\STATE{$W \leftarrow \phi$}
\FOR{\textbf{each} $z \in R_l \backslash \{x,y\}$}
\IF{$\{x,z\} \not\in E$ or $\{y,z\} \not\in E$}
\STATE{$W \leftarrow W \cup \{z\}$}
\STATE{$S_z \leftarrow S_z \backslash \{l\}$}
\ENDIF
\ENDFOR
\STATE{$R_l \leftarrow R_l \backslash W$}
\STATE{Compute and set values of indicator variables $x_l$ and $y_l$ denoting whether $x \in C_l$ and $y \in C_l$ respectively}
\RETURN $\CCS, W, x_l, y_l$
\end{algorithmic}
\end{algorithm}

\begin{algorithm}[H]
\caption{\textit{Restore}}
\label{alg:restore}
\textbf{Input:} A graph $G=(V,E)$, \emph{Clique cover} $\C$ (of a subset of edges in $E$), \emph{Candidate clique sets} $\CCS$, $l$, an edge $\{x, y\}$, $W$, $x_l$, $y_l$\\
\textbf{Output:} $C_l$ and corresponding \emph{candidate clique sets} restored
\begin{algorithmic}[1]
\FOR{\textbf{each} $z \in W$}
\STATE{$R_l \leftarrow R_l \cup \{z\}$}
\STATE{$S_z \leftarrow S_z \cup \{l\}$}
\ENDFOR
\STATE{Using indicator variables $x_l$ and $y_l$, restore containment of $x$ and $y$ in $C_l$}
\RETURN $\C, \CCS$
\end{algorithmic}
\end{algorithm}

As outlined in Algorithm \ref{alg:prepare}, \emph{Prepare} subroutine removes $l$ from the candidate clique sets of vertices in $R_l$ that cannot have $C_l$ as \emph{candidate clique} after inclusion of $\{x,y\}$ to $C_l$, and moves these vertices from $R_l$ to $W$. \emph{Prepare} subroutine also sets the values of $x_l$ and $y_l$. \emph{Restore} subroutine (Algorithm \ref{alg:restore}) simply undoes the changes made by step 7 and \emph{Prepare} subroutine, using $W$, $x_l$, and $y_l$. Similarly, \emph{RemoveNewClique} subroutine (Algorithm \ref{alg:removeCl}) undoes the modifications done to \emph{clique cover} and \emph{candidate clique sets} in \emph{AddNewClique} subroutine. Note that all four subroutines have $O(\Delta)$ time bound.

\begin{algorithm}[H]
\caption{\textit{RemoveNewClique}}
\label{alg:removeCl}
\textbf{Input:} A graph $G=(V,E)$, \emph{Clique cover} $\C$ (of a subset of edges in $E$), \emph{Candidate clique sets} $\CCS$\\
\textbf{Output:} $\C$ and $\CCS$ restored by removing the last clique added by \emph{AddNewClique}
\begin{algorithmic}[1]
\STATE{$p \leftarrow |\C|$}
\STATE{$\C \leftarrow \C \backslash C_p$}
\FOR{\textbf{each} $z \in R_p$}
\STATE{$S_z \leftarrow S_z \backslash \{p\}$}
\ENDFOR
\RETURN $\C, \CCS$
\end{algorithmic}
\end{algorithm}

Next, we show correctness, time and space bounds, depth of the search tree, and number of branches at each node of the search tree of Algorithm \ref{alg:iFPT}.

\begin{restatable}{lemma}{LemmaIfptCor}
\label{lemma_ifpt}
Algorithm \ref{alg:iFPT} correctly solves the parameterized problem \emph{ECC-d}.
\end{restatable}
\begin{proof}
Let $(G,(d,k))$ is an instance of \emph{ECC-d}. Denote Algorithm \ref{alg:iFPT} with $A_x$.

Consider a family of algorithms $A$ that consider uncovered edges $\{x,y\}$ in arbitrary order in step 4 of Algorithm \ref{alg:iFPT}. Clearly, $A_x \in A$.

Now, if $(G,(d,k))$ is a 'yes' instance, then any algorithm $A_y \in A$ will report it. This is because $A_y$ considers all possible ways to cover every uncovered edge: either $\{x,y\}$ has to be covered with one of the existing cliques $C_l$ such that $l \in S_x \cap S_y$ (step 5-13), or $\{x,y\}$ cannot be covered with any of the existing cliques (step 14-21). $A_y$ considers all such possible ways and allows search tree of depth with $k$ cliques. Similarly, if $(G,(d,k))$ is a 'no' instance, then $A_y$ returns $\phi$, denoting a 'no' instance.
\end{proof}

\begin{restatable}{lemma}{LemmaIfptSpace}
\label{lemma_ifpt_space}
In a search tree with at most $k$ cliques, Algorithm \ref{alg:iFPT} takes $O(m+k\Delta)$ space.
\end{restatable}
\begin{proof}
Let $T_i$ be a node in the search tree of Algorithm \ref{alg:iFPT} when an uncovered $\{x,y\}$ is covered with a new clique in steps 14-21 of Algorithm \ref{alg:iFPT}. In $T_i$ we must have created the new clique only because all branches of steps 5-13 have failed. Any node $T_j$ such that $T_i$ is an ancestor of $T_j$ in the search tree would not select $\{x,y\}$ as uncovered. Since, number of uncovered edge $\{x,y\}$ incident on any $x \in V$ is at most $|N(x)|$, characterization of Lemma \ref{lemma_lm_char1} holds for any \emph{clique cover} produced by Algorithm \ref{alg:iFPT}.

In the proof of Lemma \ref{lemma_csum}, since we have established that the part of the proof related to Lemma \ref{lemma_lm_char1} holds, it is straightforward to see that rest of the proof holds for any \emph{clique cover} produced by Algorithm \ref{alg:iFPT}. It follows that characterization of Lemma \ref{lemma_csum} holds for \emph{candidate clique sets} $\CCS$ used in Algorithm \ref{alg:iFPT}.

Entire $R_l$ data structures capture the same information as entire \emph{candidate clique sets} $\CCS$, i.e., $\sum |R_l| = O(m +k \Delta)$. At each node, in \emph{Prepare} subroutine, a vertex is either kept in $R_l$ or moved to $W$. Therefore, space usage of $W$ data structures in total is already accounted for, i.e., Algorithm \ref{alg:iFPT} has $O(m+k\Delta)$ space bound.
\end{proof}

\begin{restatable}{lemma}{LemmaIfptDepth}
\label{lemma_ifpt_depth}
In a search tree with at most $k$ cliques, depth of the search tree of Algorithm \ref{alg:iFPT} is at most $dk$.
\end{restatable}
\begin{proof}
Every new clique starts with two vertices (step 15), and each expansion of a clique (steps 5-13) with an uncovered edge $\{x,y\}$ will add one new vertex to the clique. Since $\forall C_l \in \C, |C_l| \leq d+1$, number of total possible expansions of a clique is $d$. Therefore, in a search tree with at most $k$ cliques, depth of the search tree is at most $dk$.
\end{proof}

\begin{restatable}{lemma}{LemmaIfptBranch}
\label{lemma_ifpt_branch}
Number of branches at any node of the search tree of Algorithm \ref{alg:iFPT} is at most $\lfloor\frac{d^2}{4}\rfloor + 1$.
\end{restatable}
\begin{proof}
Let $N_p(u_i)$ denote neighbours of $u_i$ that precede $u_i$ in the degeneracy order of $V$, i.e., $N_p(u_i) = N(u_i) \cap \{u_1, u_2, ..., u_{i-1}\} = N(u_i) \backslash N_d(u_i) $.

Invariant: At any node in the search tree of Algorithm \ref{alg:iFPT}, if $x = u_i$, then $\forall C_l \in \C, C_l \subseteq \{u_i, ..., u_n\}$.

Preceding invariant holds trivially by our choice of uncovered edge $\{x,y\}$ for each node. Since for any $x = u_i$, we only consider $y \in N_d(x)$, any $z \in N_p(x)$ will not be included in any $C_l$ until we process a node where $x = z$.

From the invariant it follows that for any $C_l \in \C$ such that $l \in S_x$, either $x \in C_l$ or ($x \not\in C_l$ and $C_l \subseteq N_d(x)$). But, for uncovered $\{x,y\}$, this must be the case that $x \not\in C_l$. Therefore, we only need to consider the case $C_l \subseteq N_d(x)$. The vertices of $N_d(x)$ can appear in most $\lfloor\frac{|N_d(x)|^2}{4}\rfloor \leq \lfloor\frac{d^2}{4}\rfloor$ cliques \cite{erdos1966representation}. It follows that $|S_x \cap S_y| \leq \lfloor\frac{d^2}{4}\rfloor$.

Therefore, number of branches in step 5-13 is bounded by $\lfloor\frac{d^2}{4}\rfloor$. Considering the branch for new clique at step 14-21, number of total branches at each node of the search tree is at most $\lfloor\frac{d^2}{4}\rfloor + 1$.
\end{proof}

\begin{restatable}{theorem}{ThmIFPT}
\label{thm_ifpt}
For \emph{ECC-d}, Algorithm \ref{alg:iFPT} is an FPT algorithm with time bound $O^{*}(d^{O(dk)})$. 
\end{restatable}
\begin{proof}
Lemma \ref{lemma_ifpt} shows the correctness proof for Algorithm \ref{alg:iFPT}. From Lemma \ref{lemma_ifpt_branch}, number of branches at any node of the search tree of Algorithm \ref{alg:iFPT} is at most $\lfloor\frac{d^2}{4}\rfloor + 1$, and from Lemma \ref{lemma_ifpt_depth}, in a search tree with at most $k$ cliques, depth of the search tree of Algorithm \ref{alg:iFPT} is at most $dk$.

It follows that in a search tree with at most $k$ cliques, total number of nodes is at most $(\lfloor\frac{d^2}{4}\rfloor + 1)^{dk}$, i.e., time bound of Algorithm \ref{alg:iFPT} is $O^{*}(d^{O(dk)})$.
\end{proof}

Each of the subroutines of Algorithm \ref{alg:iFPT} has $O(\Delta)$ time bound. If \emph{clique cover} size is $k$, considering $O(m+n)$ time to compute degeneracy ordering of $V$ and all $N_d(x)$ lists, time bound to compute a \emph{minimum clique cover} using Algorithm \ref{alg:iFPT} is $O(m+n+\Delta k(\lfloor\frac{d^2}{4}\rfloor + 1)^{dk})$.

Note that we can exploit \emph{candidate clique sets} further. For example, we can select uncovered edge arbitrarily, in step 4 of Algorithm \ref{alg:iFPT}, and still obtain an FPT algorithm with respect to parameters $k$ and $\Delta$. Applicability of \emph{candidate clique sets} may extend far beyond what we have demonstrated so far. This paradigm of computing \emph{clique cover} may give rise to better algorithms for many other objectives of \emph{clique covers}. We show an example next.

\subsubsection{\emph{Assignment-Minimum Clique Cover}.}
\label{sec:amcc}

An objective of \emph{clique cover} that has found application in problems requiring statistical data visualization is called \emph{assignment-minimum clique cover} [\cite{ennis2012assignment}, \cite{gramm2007algorithms}, \cite{piepho2004algorithm}]. The objective of this \emph{clique cover} is to find a \emph{clique cover} $C$ of $G$ such that among the \emph{clique covers} of $G$, $C$ has a minimum total number of vertices contained in the cliques of $C$, i.e., $\sum_{C_l \in \C} |C_l|$ is minimum. Prior search tree algorithm for this objective of \emph{clique cover} requires bounds exponential in input size for depth, number of branches, and space use \cite{ennis2012assignment}. Next, we describe straightforward modifications of Algorithm \ref{alg:iFPT} that solves this \emph{clique cover} problem exactly, but with polynomial bounds on depth, number of branches, and space use.

For the \emph{assignment-minimum clique cover}, we would run modifications of Algorithm \ref{alg:iFPT} with parameter $k$ set to $m$, since no \emph{clique cover} would require more than $m$ cliques. We would discard steps 9-11 and steps 17-19, since we would need to explore all \emph{clique covers} with at most $m$ cliques. Let $\C^{AM}$ denote a \emph{clique cover} with minimum total number of vertices among the \emph{clique covers} explored in the search tree. Initially, $\C^{AM}$ would be empty. In step 2 of Algorithm \ref{alg:iFPT}, while returning, we would compare $\C^{AM}$ with $\C$, and if $\C$ has smaller total number of vertices than $\C^{AM}$, then we would set $\C^{AM}$ to be $\C$. This algorithm has space bound $O(m\Delta)$ (Lemma \ref{lemma_ifpt_space}), depth bound $O(dm)$ (Lemma \ref{lemma_ifpt_depth}), and $O(d^2)$ number of branches (Lemma \ref{lemma_ifpt_branch}).

\subsection{FPT Algorithm Based on Maximal Clique Enumeration.}
\label{fpt_2}

Our maximal clique enumeration based FPT algorithm follows similar structural pattern as of the search tree algorithm of \cite{gramm2006data}, but with crucial differences. 
Using a heuristic, \cite{gramm2006data} select an uncovered edge at each node of the search tree and enumerate all maximal cliques containing the edge. The expectation is that the heuristic would reduce number of nodes in the search tree. Nevertheless, at each node, search tree algorithm of \cite{gramm2006data} searches maximal cliques in a subgraph containing at most $2^k$ vertices. We utilize degeneracy ordering to choose an edge, and restrict the search for maximal cliques in the later neighbours (in the degeneracy ordering) of an end vertex of the edge. At each node, our algorithm searches maximal cliques in a subgraph containing at most $d$ vertices.

\begin{algorithm}
\caption{\textit{MaximalCliqueFPT}}
\label{alg:mFPT}
\textbf{Input:} A graph $G=(V,E)$, \emph{Clique cover} $\C$ (of a subset of edges in $E$), an integer $k$\\
\textbf{Output:} If exists, a \emph{clique cover} $\C$ of $G$ with $|\C| \leq k$; otherwise $\phi$
\begin{algorithmic}[1]
\IF{$\C$ covers $G$}
\RETURN $\C$
\ENDIF
\IF{$k \leq 0$}
\RETURN $\phi$
\ENDIF
\STATE{Select an uncovered edge $\{x,y\}$ such that $y \in N_d(x)$, and in the degeneracy order of $V$, $x$ precedes any other vertex with such a neighbour $y$}
\FOR{\textbf{each} maximal clique $R$ containing $\{x,y\}$ in the subgraph induced by $(N_d(x) \cap N(y)) \cup \{x,y\}$}
\STATE{$Q \leftarrow MaximalCliqueFPT(G, \C \cup \{R\}, k-1)$}
\IF{$Q \ne \phi$}
\RETURN $Q$
\ENDIF
\ENDFOR
\RETURN $\phi$
\end{algorithmic}
\end{algorithm}

Algorithm \ref{alg:mFPT} outlines our maximal clique enumeration based FPT algorithm that solves \emph{ECC-d}. In step 7 of Algorithm \ref{alg:mFPT}, we select an uncovered edge that helps us to bound number of maximal cliques need to be enumerated at each node of the search tree. More specifically, we select an edge $\{x,y\}$ such that $y \in N_d(x)$ and all neighbours of $x$ that precedes $x$ in the degeneracy order of $V$ do need not to be considered in the search of maximal cliques for $\{x,y\}$. Finally, we branch into each of the choices of maximal cliques in steps 8-13. We obtain following results.

\begin{restatable}{lemma}{LemmaMfptCor}
\label{lemma_mfpt}
Algorithm \ref{alg:mFPT} correctly solves the parameterized problem \emph{ECC-d}.
\end{restatable}
\begin{proof}
Let $N_p(u_i)$ denote neighbours of $u_i$ that precede $u_i$ in the degeneracy order of $V$, i.e., $N_p(u_i) = N(u_i) \cap \{u_1, u_2, ..., u_{i-1}\} = N(u_i) \backslash N_d(u_i) $. Let $(G,(d,k))$ is an instance of \emph{ECC-d}. Denote Algorithm \ref{alg:mFPT} with $A_x$.

Consider a family of algorithms $B$ that consider uncovered edges $\{x,y\}$ in arbitrary order in step 7 of Algorithm \ref{alg:mFPT}, and in steps 8-13 branches into each of the choices of maximal cliques containing $\{x,y\}$ in the subgraph induced by $(N(x) \cap N(y)) \cup \{x,y\}$.

Now, if $(G,(d,k))$ is a 'yes' instance, then any algorithm $B_y \in B$ will report it. This is because $B_y$ considers all possible ways to cover every uncovered edge: $\{x,y\}$ can be covered with one of the maximal cliques containing $\{x,y\}$ in the subgraph induced by $(N(x) \cap N(y)) \cup \{x,y\}$. $B_y$ considers all such possible maximal cliques, and allows search tree of depth with $k$ cliques. Similarly, if $(G,(d,k))$ is a 'no' instance, then $B_y$ returns $\phi$, denoting a 'no' instance.

$A_x$ deviates from the family of algorithms $B$: $A_x$ considers maximal cliques containing $\{x,y\}$ in the subgraph induced by $(N_d(x) \cap N(y)) \cup \{x,y\}$. Next, we show that this deviation is safe.

Invariant: At any node in the search tree of Algorithm \ref{alg:mFPT}, if $x = u_i$, then all edges incident on $\{u_1, u_2..., u_{i-1}\}$ are covered.

Preceding invariant holds trivially, by our choice of uncovered edge $\{x,y\}$, for each node in the search tree.

Note that $y \in N_d(x) \subseteq \{u_{i+1}, ..., u_n\}$ and $N_d(x) \cap N(y) \subseteq \{u_{i+1}, ..., u_n\}$. Now, the vertices that are not considered by $A_x$ in the enumeration of maximal cliques are $(N(x) \cap N(y)) \backslash (N_d(x) \cap N(y)) = N_p(x) \cap N(y) \subseteq \{u_1, ..., u_{i-1}\}$. By the invariant, all edges incident on $\{u_1, ...., u_{i-1}\}$ are covered when we are considering maximal cliques for $\{x, y\}$. Therefore, it is safe to consider only maximal cliques in the subgraph induced by $(N_d(x) \cap N(y)) \cup \{x,y\}$.
\end{proof}

Using the Moon-Moser bound, we get desired time bound for Algorithm \ref{alg:mFPT}.

\begin{restatable}{theorem}{ThmMFPT}
\label{thm_mfpt}
For \emph{ECC-d}, Algorithm \ref{alg:mFPT} is an FPT algorithm with time bound $O^{*}(3^{O(dk)})$. 
\end{restatable}
\begin{proof}
Lemma \ref{lemma_mfpt} shows the correctness proof for Algorithm \ref{alg:mFPT}. Next, we show the fixed-parameter tractability of Algorithm \ref{alg:mFPT}.

Since $|N_d(x) \cap N(y)| \leq d$, by the Moon-Moser bound, number of maximal cliques in a subgraph induced by vertices of $N_d(x) \cap N(y)$ is at most $3^{\frac{d}{3}}$, i.e., number of maximal cliques enumerated at each node of the search tree is at most $3^{\frac{d}{3}}$. Therefore, total number of nodes in the search tree is at most $3^\frac{dk}{3}$, i.e., time bound of Algorithm \ref{alg:mFPT} is $O^{*}(3^{O(dk)})$.
\end{proof}

At each node of the search tree of Algorithm \ref{alg:mFPT}, we need to enumerate maximal cliques of a subgraph containing $|N_d(x) \cap N(y)| = O(d)$ vertices. Assuming maximal clique enumeration algorithm uses space proportional to the size of input graph, in a search tree with at most $k$ cliques, space use of Algorithm \ref{alg:mFPT} is $O(m+d^2k)$.

Using the Moon-Moser bound for the search tree algorithm of \cite{gramm2006data}, we obtain $O^{*}(3^{O(k2^k)})$ time bound. It follows that for $2^k>d$, Algorithm \ref{alg:mFPT} improves upon the time bound of the search tree algorithm of \cite{gramm2006data} by a factor of $3^{k(2^k-d)}$.

We can also compare the algorithms using time bounds of maximal clique enumeration algorithms. Using the time bound of \cite{eppstein2013listing}, we get $O^*{}(d^k3^{O(dk)})$ and $O^{*}(2^{k^2}3^{O(dk)})$ time bounds for Algorithm \ref{alg:mFPT} and search tree algorithm of \cite{gramm2006data} respectively. For $2^k>d$, this results in $(\frac{2^k}{d})^k$ factor improvement of time bound by Algorithm \ref{alg:mFPT} on the time bound of the search tree algorithm of \cite{gramm2006data}.

Note that our FPT algorithm based on maximal clique enumeration (Algorithm \ref{alg:mFPT}) has better dependency on $d$ in the time bound than our FPT algorithm based on \emph{candidate clique sets} (Algorithm \ref{alg:iFPT}): $O^{*}(3^{O(dk)})$ vs. $O^{*}(d^{O(dk)})$. Therefore, we expect Algorithm \ref{alg:mFPT} to perform better than Algorithm \ref{alg:iFPT} on dense graphs.
\section{Experimental Analyses}
\label{sec:exps}

In this section we present a summary of our experimental analyses of our FPT algorithms (Section \ref{exp_fpt}) and greedy framework (Section \ref{exp_greedy}). We compare runtimes and other metrics of our algorithms against the state of the art, on random graphs, as well as on real-world graphs. We ask reader to consult Appendix \ref{sec:app_exp} for further details of the results discussed in this section.

We ran our experiments on an AMD Epyc 7662 processors based system (part of Purdue University Community Cluster) called Bell\footnote{\url{rcac.purdue.edu/compute/bell}}. Nodes in the cluster has 128 cores, each core running at 2.0 GHz, with 16 MB unified L3 Cache. Bell nodes have 256-1024 GB memory, and run CentOS 7 operating system. We implemented our algorithms in the C++ programming language and compiled using gcc compiler (version 9.3.0) with the -O3 optimization flag. Our implementations are freely available upon request (will be available at \url{github.com}).

\subsection{Experimental Analysis of FPT algorithms}
\label{exp_fpt}

\subsubsection{Dataset.}
In this analysis we evaluated all our algorithms using random graphs from the $G(n,p)$ model. Varying $n$ (150-1000) and $p$ (0.035 - 0.11), we generated 16 ensembles of $G(n,p)$, each containing 32 random graphs. Appendix \ref{sec:app_exp_dataset1} contains details of the dataset and related statistics.

\subsubsection{Algorithms Tested.}
We have found that prepossessing with data reduction rules such as presented by \cite{gramm2006data} is crucial for the search tree algorithms. In our study we have not found any graph instance which has application of interleaving data reduction rules in the search tree (as done in \cite{gramm2006data}). Therefore, we used data reduction rules of \cite{gramm2006data} as prepossessing steps for our FPT algorithms. We implemented and tested four algorithms: two of our FPT algorithms outlined in Algorithm \ref{alg:iFPT} (\emph{CFPT}) and Algorithm \ref{alg:mFPT} (\emph{MFPT}), the FPT algorithm of \cite{gramm2006data} (\emph{GFPT}), and the FPT algorithm of \cite{gramm2006data} without interleaving the data reduction rules (\emph{GFPT-NI}). We evaluated \emph{MFPT} with three different maximal clique enumeration algorithms: \cite{koch2001enumerating} (\emph{Koch}), \cite{tomita2006worst} (\emph{TTT}), \cite{eppstein2013listing} (\emph{ELS}); and based on our evaluation, we used \emph{ELS} based implementation for \emph{MFPT}, \emph{GFPT}, and \emph{GFPT-NI} (Appendix \ref{sec:app_exp_mce}).

\subsubsection{Experimental Results.}
All tested algorithms were allowed to run up to 4 hours on each of the 512 graphs. Let \emph{completion rate} denote percentage of number of graphs (in each of the graph ensembles) for which \emph{clique covers} were computed within 4 hours. \emph{Completion rates} are significantly small from \emph{GFPT} (18.75\%-46.88\%) and \emph{GFPT-NI} (31.25\%-68.75\%) (Appendix \ref{sec:app_exp_fptCompletion}). Both \emph{CFPT} (\emph{completion rates} 93.75\% -100\%) and \emph{MFPT} (\emph{completion rates} 96.88\% -100\%) computed \emph{clique cover} of all 32 graphs of most of the graph ensembles.

\begin{table}
\caption{Average runtimes of \emph{GFPT}, \emph{CFPT}, and \emph{MFPT}. Each of the row statistics is on the instances on which \emph{GFPT} terminated within 4 hours (Appendix \ref{sec:app_exp_fptCompletion}).}
\centering
\begin{tabular}{lrrrr}
\hline
\multirow{2}{*}{Graph} & \multicolumn{3}{c}{Average Runtime   (seconds)} \\
\cline{2-4}
                       &\multicolumn{1}{c}{\emph{GFPT}}            & \multicolumn{1}{c}{\emph{CFPT}}          & \multicolumn{1}{c}{\emph{MFPT}}          \\
\hline
er\_150\_0100          & 3123.58         & 8.62          & 7.90          \\
er\_150\_0102          & 3545.24         & 8.98          & 22.11         \\
er\_150\_0104          & 2939.30         & 12.56         & 5.21          \\
er\_150\_0106          & 4121.85         & 23.70         & 22.75         \\
er\_150\_0108          & 4071.08         & 2.64          & 9.97          \\
er\_150\_0110          & 4117.19         & 0.61          & 7.14          \\
                       &                 &               &               \\
er\_250\_0070          & 2570.21         & 13.66         & 4.59          \\
er\_260\_0070          & 3190.58         & 2.14          & 3.45          \\
er\_270\_0070          & 3724.74         & 5.06          & 9.04          \\
er\_280\_0070          & 2001.80         & 8.56          & 11.10         \\
er\_290\_0070          & 3327.44         & 13.34         & 26.33         \\
er\_300\_0070          & 3921.95         & 2.21          & 9.98          \\
                       &                 &               &               \\
er\_200\_009           & 4437.79         & 6.24          & 12.18         \\
er\_500\_005           & 4425.95         & 1.99          & 7.39          \\
er\_750\_0040          & 5227.85         & 1.82          & 5.18          \\
er\_1000\_0035         & 3405.49         & 2.80          & 6.20 \\
\hline
\end{tabular}
\label{tab:fptRunTime}
\end{table}

Table \ref{tab:fptRunTime} shows the runtime comparisons of \emph{GFPT}, \emph{CFPT}, and \emph{MFPT}. Both of our algorithms show significantly better performance than \emph{GFPT}: \emph{CFPT} is approximately 100-5000 times faster than \emph{GFPT}, and \emph{MFPT} is approximately 250-900 times faster than \emph{GFPT} (see Appendix \ref{sec:app_exp_fptRuntime} for geometric mean of ratios of runtimes).

\emph{GFPT-NI} showed better performance than \emph{GFPT} on some of the graph ensembles (Table \ref{tab:fptRunTime2}, Appendix \ref{sec:app_exp_fptRuntime}). \emph{GFPT} and \emph{GFPT-NI} have the same search tree for all our test graphs, since in our study we have not found any application of data reduction rules during the execution of the search tree algorithm of \cite{gramm2006data}. The reduction of runtime of \emph{GFPT-NI}, compared to \emph{GFPT}, comes from the time saved by not attempting to apply data reduction rules at each node of the search tree (for a graph with $n$ vertices, data reduction rules have $O(n^4$) time bound \cite{gramm2006data}). But, even with this savings in runtimes, performance of \emph{GFPT-NI} are nowhere close to the performances of our algorithms.

The reason for stunning performances of our algorithms is in the fact that the search trees generated by our algorithms are remarkably small. Table \ref{tab:fptSearchTree} shows comparison results of number of nodes in the search trees of \emph{GFPT} and \emph{MFPT}. Search trees of \emph{GFPT} (or \emph{GFPT-NI}) and \emph{MFPT} are directly comparable: every node in the search tree corresponds to a new clique, added to \emph{clique cover}. Search trees of \emph{MFPT} are approximately 50-150 times smaller than search trees of \emph{GFPT} (or \emph{GFPT-NI}). Note that these massive reductions of number of nodes in the search trees are only on the graphs on which \emph{GFPT} terminated within 4 hours. With increased time cut off ($>$ 4 hours), these reductions would get even larger: both \emph{CFPT} and \emph{MFPT} can solve graph instances in minutes, which \emph{GFPT} could not solve in 24 hours.

Search trees of \emph{CFPT} are not directly comparable with search trees of other algorithms: search tree of \emph{CFPT} has depth bound $dk$ (Lemma \ref{lemma_ifpt_depth}), whereas search trees of other algorithms have depth bound $k$. Nevertheless, \emph{CFPT} shows even better performance than \emph{MFPT} on most of the graph ensembles, despite the fact that \emph{MFPT} has better dependency on $d$ in the time bound than \emph{CFPT}: $O^{*}(3^{O(dk)})$ vs. $O^{*}(d^{O(dk)})$.

\begin{table}
\caption{Average number of nodes (in million) generated in the search trees of \emph{GFPT} (or \emph{GFPT-NI}) and \emph{MFPT}. Last column shows geometric mean of ratios of number of nodes in the search tree of \emph{GFPT} to number of nodes in the search tree of \emph{MFPT}.}
\centering
\begin{tabular}{lrrr}
\hline
\multirow{4}{*}{Graph} & \multicolumn{2}{c}{Average number of} & \multicolumn{1}{c}{Geometric mean} \\
                       & \multicolumn{2}{c}{search tree nodes} & \multicolumn{1}{c}{of ratios of}  \\
                       & \multicolumn{2}{c}{(in million)}      & \multicolumn{1}{c}{number of search}         \\
                       \cline{2-3}
                       & \multicolumn{1}{c}{\emph{GFPT}} & \multicolumn{1}{c}{\emph{MFPT}}             & \multicolumn{1}{c}{tree nodes}        \\
\hline
er\_150\_0100          & 2387.42            & 36.16            & 82.90             \\
er\_150\_0102          & 4542.24            & 77.11            & 82.51             \\
er\_150\_0104          & 4702.10            & 86.70            & 86.52             \\
er\_150\_0106          & 3266.89            & 36.06            & 111.32            \\
er\_150\_0108          & 3730.83            & 36.09            & 137.56            \\
er\_150\_0110          & 6287.48            & 111.71           & 85.79             \\
                       &                    &                  &                   \\
er\_250\_0070          & 2335.21            & 14.83            & 98.83             \\
er\_260\_0070          & 5622.55            & 27.29            & 144.26            \\
er\_270\_0070          & 4440.97            & 36.18            & 111.81            \\
er\_280\_0070          & 781.74             & 18.62            & 44.37             \\
er\_290\_0070          & 3116.29            & 56.33            & 58.73             \\
er\_300\_0070          & 2075.01            & 26.86            & 64.44             \\
                       &                    &                  &                   \\
er\_200\_009           & 4443.24            & 105.26           & 107.83            \\
er\_500\_005           & 4410.74            & 24.63            & 153.69            \\
er\_750\_0040          & 3649.86            & 67.32            & 126.41            \\
er\_1000\_0035         & 3400.76            & 27.74            & 94.64            \\
\hline
\end{tabular}
\label{tab:fptSearchTree}
\end{table}
\subsection{Experimental Analysis of Greedy Framework}
\label{exp_greedy}

\subsubsection{Dataset.}
Our dataset for this analysis consists of 30 real-world graphs with number of edges ranging from approximately 1 million to 1.8 billion (Appendix \ref{sec:app_exp_dataset2}). The dataset contains 21 brain networks from \url{networkrepository.com} \cite{nr}, 6 social networks from the Stanford Large Network Dataset Collection \cite{snapnets}, and 3 gene regulatory networks from the Suite Sparse Matrix Collection (\url{sparse.tamu.edu} \cite{davis2011university}).

\subsubsection{Algorithms Tested.}
Our graphs of this analysis are not solvable by the FPT algorithms in a reasonable amount of time (24 hours). Therefore, we consider fast variants of the state of the art heuristic algorithms. We tested four algorithms: improvement of the Kellerman heuristic \cite{kellerman1973determination} by \cite{gramm2006data} (\emph{g-ALG}), the \emph{ECC-rc} variant of heuristic algorithm by \cite{conte2020large} (\emph{c-ALG}), two of our algorithms, outlined in Algorithm \ref{alg:dine} (\emph{CCSG}) and modifications of \emph{CCSG} for degeneracy ordering outlined in Section \ref{sec:degen} (\emph{CCSD}). For \emph{CCSG} (also for \emph{g-ALG}), we used degree based edge ordering (edges incident on low degree vertices are processed first), and for both \emph{CCSG} and \emph{CCSD} we selected largest clique in step 5 of Algorithm \ref{alg:dine}.

\subsubsection{Metrics of Comparison.}
Our metric to compare quality of \emph{clique cover} is relative reduction of \emph{clique cover size}. For relative reduction, we do not consider \emph{trivial cliques} \footnote{A \emph{trivial clique} is an edge whose end vertices do not share any common vertex in their neighbourhood. Any algorithm must cover every such edge by a clique containing only the edge.} of a graph in the \emph{clique cover}, since \emph{trivial cliques} are common among all \emph{clique covers}. Therefore, if number of cliques (\emph{without trivial cliques}) produced by algorithm 'A', and algorithm 'B' are 'a', and 'b', respectively, then relative reduction of \emph{clique cover} size by 'B' (relative to 'A') is (1 - $\frac{b}{a}) \times 100$.

For the performance comparison we use relative performance: if 'a' is the runtime of algorithm 'A', and 'b' (in same unit as 'a') is the runtime of algorithm 'B', then relative performance of 'B' (relative to 'A') is the ratio of 'a' to 'b'.

\subsubsection{Experimental Results.}
All tested algorithms were allowed to run up to 24 hours on each of the graphs of the dataset. Out of 30 graphs (Appendix \ref{sec:app_exp_dataset2}), \emph{c-ALG} terminated on 29 graphs, and \emph{g-ALG} terminated only on 13 graphs. Our algorithms terminated on all the graphs within 3 hours: longest time taken by \emph{CCSG} is 2.57 hours, and by \emph{CCSD} is 2.26 hours.

\begin{table}
\centering
\caption{Average relative reduction of \emph{clique cover} size by  \textit{CCSD} and \emph{CCSG} (relative to \textit{g-ALG} and \textit{c-ALG}).}
\begin{tabular}{lrr}
\hline
Algorithm & \multicolumn{2}{l}{Average relative   reduction} \\
\cline{2-3}
compared  & CCSG                    & CCSD                   \\
\hline
g-ALG                                 & 27.14                   & 30.15                  \\
c-ALG                                 & 13.15                   & 16.36 \\
\hline
\end{tabular}
\label{tab:comparison}
\end{table}

\begin{table}
\centering
\caption{Geometric mean of relative performances of \textit{CCSD} and \emph{CCSG} (relative to \textit{g-ALG} and \textit{c-ALG}).}
\begin{tabular}{lrr}
\hline
Algorithm & \multicolumn{2}{l}{Mean relative performance} \\
\cline{2-3}
compared  & CCSG                    & CCSD                   \\
\hline
g-ALG                                 & 80.22                         & 93.91                        \\
c-ALG                                 & 2.51                          & 2.91                         \\
\hline
\end{tabular}
\label{tab:comparison_time}
\end{table}

For the graphs on which \emph{g-ALG} and \emph{c-ALG} terminated within 24 hours, Table \ref{tab:comparison} and Table \ref{tab:comparison_time} summarize comparison of \emph{CCSG} and \emph{CCSD} against \emph{g-ALG} and \emph{c-ALG} on the two metrics (details are in Appendix \ref{sec:app_exp_greedyResults}). On both metrics, both of our algorithms outperformed \emph{g-ALG} and \emph{c-ALG}. Notably, our algorithms are approximately 90 times faster on average than \emph{g-ALG}, and \emph{clique cover} size of our algorithms are approximately 30\% smaller on average compared to \emph{g-ALG}.

\begin{figure}[H]
\centering

\begin{tikzpicture}[scale=1.0]
\begin{axis}[
    xlabel={$m$},
    ylabel={$CCS_{max}$},
    ymajorgrids=true,
    grid style=dashed,
    legend pos=north west,
    xmode=log,
    ymode=log,
    ]

\addplot [only marks, color=blue, mark=star] table {data_plots/CCS_bound1/ccsg.dat};

\addplot [color=brown] table {data_plots/CCS_bound1/m1.dat};

\addplot [color=orange] table {data_plots/CCS_bound1/m2.dat};

\addplot [color=red] table {data_plots/CCS_bound1/m3.dat};

\legend{\textit{CCSG}, $m$, $2m$, $3m$}

\end{axis}
\end{tikzpicture}
\caption{$CCS_{max} =max\{\sum_{x\in V}|S_x|\}$ vs. $m$ for all 30 test graphs of Table \ref{tab:dataset}.}
\label{fig:dine_ms}
\end{figure}
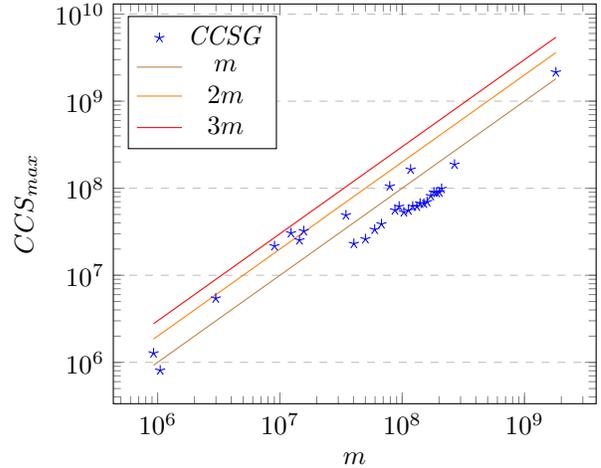

All tested algorithms used same post-processing step (by \cite{kou1978covering}). For the entire dataset, average and maximum relative reduction of \emph{clique cover} size by post-processing in \textit{CCSG} are only 0.05 and 0.1 respectively: this gives strong experimental verification of extremely rare redundancy in \emph{locally minimal clique cover} (Lemma \ref{lemma_lm_char4}). Also, for the graphs of Appendix \ref{sec:app_exp_dataset1}, \emph{clique cover} sizes from our greedy algorithms are only 1\% to 2\% larger on average than the optimal \emph{clique cover} sizes; but our greedy algorithms are orders of magnitude faster than the FPT algorithms (Appendix \ref{sec:app_exp_greedyFPT}).

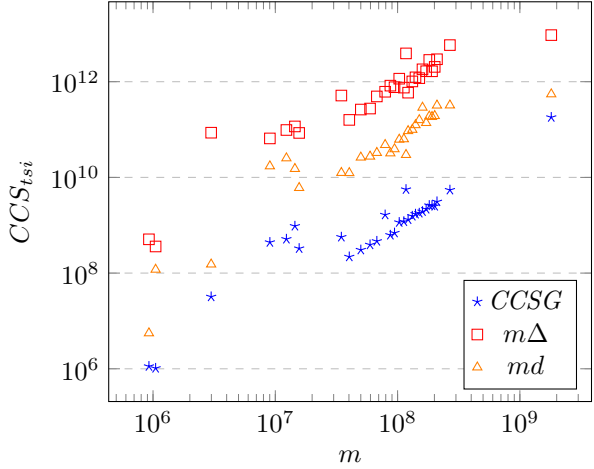
\begin{figure}[!t]
\centering

\begin{tikzpicture}[scale=1.0]
\begin{axis}[
    xlabel={$m$},
    ylabel={$CCS_{tsi}$},
    ymajorgrids=true,
    grid style=dashed,
    legend pos=south east,
    xmode=log,
    ymode=log,
    ]

\addplot [only marks, color=blue, mark=star] table {data_plots/CCS_bound2/ccsg.dat};

\addplot [only marks, color=red, mark=square] table {data_plots/CCS_bound2/mdelta.dat};

\addplot [only marks, color=orange, mark=triangle] table {data_plots/CCS_bound2/md.dat};

\legend{\textit{CCSG}, $m\Delta$, $md$}

\end{axis}
\end{tikzpicture}
\caption{Total size of the intersections at step 4 of Algorithm \ref{alg:dine} ($CCS_{tsi}$) vs. $m$ for all 30 test graphs of Table \ref{tab:dataset}. Both $m\Delta$ and $md$ values show that they are consistent upper bounds of $CCS_{tsi}$.}
\label{fig:dine_ts}
\end{figure}

In Theorem \ref{thm_ccsg} and Theorem \ref{thm_ccsd} we have described results for the cases when the bound $\sum_{x \in V} |S_x| = O(m)$ holds. We have found that the bound holds for all our test graphs. Algorithm \ref{alg:dine} grows and shrinks the value of the $\sum_{x \in V} |S_x|$. In Figure \ref{fig:dine_ms}, we show the maximum value of the sum $CCS_{max} = max\{\sum_{x \in V} |S_x|\}$ in the entire execution of the algorithm. Also, in Figure \ref{fig:dine_ts}, we plotted total size of the intersections $CCS_{tsi}$ at step 4 of Algorithm \ref{alg:dine}, which is bounded by $\sum_{\{x,y\}\in E} min\{|S_x|, |S_y|\}$. $CCS_{tsi}$ is the dominating time of Algorithm \ref{alg:dine}. To avoid cluttering the plots, we omitted similar results of \emph{CCSD}.

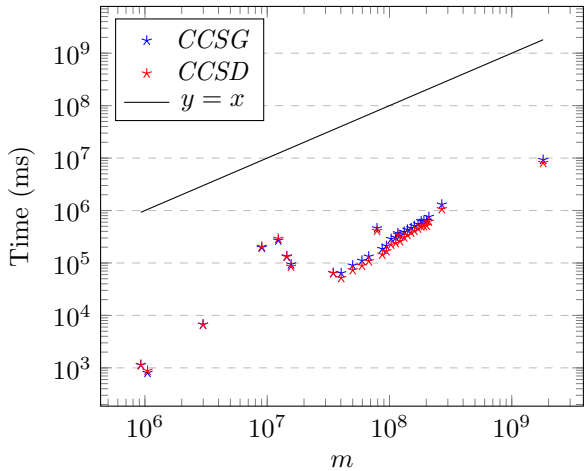
\begin{figure}[!t]
\centering

    \begin{tikzpicture}
    \begin{axis}[
        xlabel={$m$},
        ylabel={Time (ms)},
        ymajorgrids=true,
        grid style=dashed,
        legend pos=north west,
        xmode=log,
        ymode=log,
        ]
    
    \addplot [only marks, color=blue, mark=star] table {data_plots/linearTime/ccsg.dat};
    
    \addplot [only marks, color=red, mark=star] table {data_plots/linearTime/ccsd.dat};
    
    \addplot [color=black] table {data_plots/linearTime/y_x.dat};
    
    \legend{\textit{CCSG}, \textit{CCSD}, $y = x$}
    
    \end{axis}
    \end{tikzpicture}
\caption{Runtimes (milliseconds) of \emph{CCSG} and \emph{CCSD} vs. $m$ for all 30 real-world graphs of the dataset.}
\label{fig:linear_time}
\end{figure}

In Figure \ref{fig:linear_time}, we show runtimes (milliseconds) of \emph{CCSG} and \emph{CCSD} against number of edges of graph. We conclude that for the test graphs of Table \ref{tab:dataset}, both \emph{CCSG} and \emph{CCSD} exhibit linear time bound.
\section{Conclusions}
\label{sec:conclusions}

We have presented a set theoretic concept and its use as a design paradigm for computing different objectives of \emph{clique covers}. We have shown that integration of structural parameters of graph along with natural parameter can lead to design of better algorithms for \emph{clique cover} problems, both in asymptotic time bound and in practice. We have provided detailed comparison results of different metrics of our algorithms against the state of the art, substantiating excellent performance of our algorithms.

Motivated by real-world applications, recently, a number of \emph{clique cover} problems have been studied \cite{cooley2021parameterized} that consider edge weights of graph. We believe (alluding to the polynomial bounds we have shown for the \emph{assignment-minimum clique cover} \cite{ennis2012assignment}) our set theoretic concept and its use in utilizing structural parameters of graph would prompt further studies of these edge weighted \emph{clique cover} problems.

\bibliographystyle{plain}
\bibliography{8_biblio.bib}

\clearpage
\appendix
\onecolumn
\section{Additional Details of Experimental Analyses}
\label{sec:app_exp}

\subsection{Dataset for the analysis of Section \ref{exp_fpt}.}
\label{sec:app_exp_dataset1}
In our preliminary experiments, we found that many random graph instances that are solvable by \emph{CFPT} and \emph{MFPT} in minutes are not solvable by \emph{GFPT} or \emph{GFPT-NI} in 24 hours. Therefore, by extensive testing we selected $n$ and $p$ such that $G(n,p)$ instances are likely to be solvable by \emph{GFPT} in 4 hours. These graphs are chosen so that search trees of \emph{GFPT} are not too small for comparison, but not too large for 4 hours time cut off. For each fixed $n$ and $p$ of Table \ref{tab:datasetFPT}, we generated as many random graphs as needed to find 32 random graphs that have number of uncovered edges between 45 and 70, after the application of data reduction rules of \cite{gramm2006data}.

\begin{table}[H]
\caption{Our test graphs from the $G(n,p)$ model. For fixed $n$ and $p$, we collected 32 $G(n,p)$ instances, i.e., each row stands for an ensemble of 32 random graphs.}
\begin{minipage}{\textwidth}
\centering
\begin{tabular}{lrrrrr}
\hline
\multicolumn{1}{c}{Graph}      & \multicolumn{1}{c}{$n$} & \multicolumn{1}{c}{$p$}             & 
\multicolumn{1}{c}{$m$ (average)} &
\multicolumn{1}{c}{$d$ (average)}  & 
\multicolumn{1}{c}{$\Delta$ (average)}
\\
\hline
er\_150\_0100  & 150  & 0.100 & 1126.22  & 10.34 & 25.44 \\
er\_150\_0102  & 150  & 0.102 & 1147.81  & 10.63 & 25.38 \\
er\_150\_0104  & 150  & 0.104 & 1163.41  & 10.84 & 26.38 \\
er\_150\_0106  & 150  & 0.106 & 1168.22  & 10.75 & 26.25 \\
er\_150\_0108  & 150  & 0.108 & 1183.09  & 11.06 & 26.25 \\
er\_150\_0110  & 150  & 0.110 & 1196.81  & 11.06 & 26.19 \\
er\_250\_007   & 250  & 0.070 & 2218.44  & 12.28 & 30.91 \\
er\_260\_007   & 260  & 0.070 & 2393.63  & 12.91 & 31.66 \\
er\_270\_007   & 270  & 0.070 & 2564.31  & 13.22 & 31.78 \\
er\_280\_007   & 280  & 0.070 & 2739.75  & 13.78 & 32.25 \\
er\_290\_007   & 290  & 0.070 & 2921.41  & 14.16 & 32.91 \\
er\_300\_007   & 300  & 0.070 & 3110.97  & 14.81 & 34.13 \\
er\_200\_009   & 200  & 0.090 & 1765.94  & 12.31 & 29.31 \\
er\_500\_005   & 500  & 0.050 & 6233.88  & 17.84 & 40.97 \\
er\_750\_004   & 750  & 0.040 & 11238.00 & 21.56 & 48.31 \\
er\_1000\_0035 & 1000 & 0.035 & 17413.44 & 25.63 & 53.84 \\
\hline
\end{tabular}
\end{minipage}
\label{tab:datasetFPT}
\end{table}

\subsection{\emph{Completion Rates} of FPT algorithms.}
\label{sec:app_exp_fptCompletion}
Table \ref{tab:fptCompletion} shows percentage of number of graphs (in each of the graph ensembles) for which \emph{clique covers} were computed within 4 hours by the FPT algorithms.

\begin{table}[!ht]
\caption{Completion rates of the FPT algorithms in 4 hours on the 32 graphs of each of the graph ensembles of Table \ref{tab:datasetFPT}.}
\begin{minipage}{\textwidth}
\centering
\begin{tabular}{lrrrr}
\hline
\multicolumn{1}{c}{Graph}  & \multicolumn{1}{c}{\emph{GFPT}}             & 
\multicolumn{1}{c}{\emph{GFPT-NI}} &
\multicolumn{1}{c}{\emph{CFPT}}  & 
\multicolumn{1}{c}{\emph{MFPT}}
\\
\hline
er\_150\_0100  & 46.88 & 62.50 & 96.88  & 100.00 \\
er\_150\_0102  & 43.75 & 65.63 & 100.00 & 100.00 \\
er\_150\_0104  & 31.25 & 53.13 & 93.75  & 96.88  \\
er\_150\_0106  & 43.75 & 62.50 & 93.75  & 96.88  \\
er\_150\_0108  & 31.25 & 50.00 & 100.00 & 100.00 \\
er\_150\_0110  & 18.75 & 31.25 & 96.88  & 100.00 \\
er\_250\_007   & 40.63 & 46.88 & 100.00 & 96.88  \\
er\_260\_007   & 34.38 & 56.25 & 100.00 & 96.88  \\
er\_270\_007   & 28.13 & 50.00 & 100.00 & 100.00 \\
er\_280\_007   & 43.75 & 46.88 & 100.00 & 100.00 \\
er\_290\_007   & 31.25 & 46.88 & 100.00 & 100.00 \\
er\_300\_007   & 37.50 & 43.75 & 100.00 & 96.88  \\
er\_200\_009   & 34.38 & 53.13 & 100.00 & 96.88  \\
er\_500\_005   & 31.25 & 50.00 & 100.00 & 100.00 \\
er\_750\_004   & 28.13 & 46.88 & 100.00 & 96.88  \\
er\_1000\_0035 & 37.50 & 68.75 & 100.00 & 96.88 \\
\hline
\end{tabular}
\end{minipage}
\label{tab:fptCompletion}
\end{table}

\subsection{Runtimes of FPT algorithms.}
\label{sec:app_exp_fptRuntime}

Table \ref{tab:fptRunTime1} and Table \ref{tab:fptRunTime2}. Each of the row statistics is on the instances on which \emph{GFPT} or \emph{GFPT-NI} terminated in 4 hours (see Table \ref{tab:fptCompletion} for completion rate).

\begin{table}[H]
\caption{Average runtime (in seconds) of \emph{GFPT}, \emph{CFPT}, and \emph{MFPT}. Fifth and sixth columns show geometric mean of ratios of runtime of \emph{GFPT} to runtime of \emph{CFPT} and \emph{MFPT} respectively.}
\begin{minipage}{\textwidth}
\centering
\begin{tabular}{@{\extracolsep{6pt}}lrrrrr@{}}
\hline
\multirow{2}{*}{Graph} & \multicolumn{3}{l}{Average Runtime   (seconds)} & \multicolumn{2}{l}{Geometric mean of   ratios of runtimes} \\
\cline{2-4} \cline{5-6}
                       & \emph{GFPT}            & \emph{CFPT}          & \emph{MFPT}          & \emph{CFPT}                         & \emph{MFPT}                        \\
\hline
er\_150\_0100          & 3123.58         & 8.62          & 7.90          & 459.80                       & 404.41                      \\
er\_150\_0102          & 3545.24         & 8.98          & 22.11         & 496.08                       & 316.72                      \\
er\_150\_0104          & 2939.30         & 12.56         & 5.21          & 157.78                       & 370.58                      \\
er\_150\_0106          & 4121.85         & 23.70         & 22.75         & 461.25                       & 408.18                      \\
er\_150\_0108          & 4071.08         & 2.64          & 9.97          & 1267.09                      & 567.08                      \\
er\_150\_0110          & 4117.19         & 0.61          & 7.14          & 4814.52                      & 519.39                      \\
 & & & & & \\
er\_250\_0070          & 2570.21         & 13.66         & 4.59          & 379.88                       & 539.54                      \\
er\_260\_0070          & 3190.58         & 2.14          & 3.45          & 1323.85                      & 547.84                      \\
er\_270\_0070          & 3724.74         & 5.06          & 9.04          & 500.22                       & 304.70                      \\
er\_280\_0070          & 2001.80         & 8.56          & 11.10         & 438.50                       & 240.49                      \\
er\_290\_0070          & 3327.44         & 13.34         & 26.33         & 496.22                       & 249.87                      \\
er\_300\_0070          & 3921.95         & 2.21          & 9.98          & 3395.05                      & 408.77                      \\
 & & & & & \\
er\_200\_009           & 4437.79         & 6.24          & 12.18         & 1458.80                      & 406.40                      \\
er\_500\_005           & 4425.95         & 1.99          & 7.39          & 1501.51                      & 645.31                      \\
er\_750\_0040          & 5227.85         & 1.82          & 5.18          & 2376.71                      & 877.33                      \\
er\_1000\_0035         & 3405.49         & 2.80          & 6.20          & 906.00                       & 452.02            \\
\hline
\end{tabular}
\end{minipage}
\label{tab:fptRunTime1}
\end{table}

\begin{table}[H]
\caption{Average runtime (in seconds) of \emph{GFPT-NI}, \emph{CFPT}, and \emph{MFPT}. Fifth and sixth columns show geometric mean of ratios of runtime of \emph{GFPT-NI} to runtime of \emph{CFPT} and \emph{MFPT} respectively.}
\begin{minipage}{\textwidth}
\centering
\begin{tabular}{@{\extracolsep{6pt}}lrrrrr@{}}
\hline
\multirow{2}{*}{Graph} & \multicolumn{3}{l}{Average Runtime   (seconds)} & \multicolumn{2}{l}{Geometric mean of   ratios of runtimes} \\
\cline{2-4} \cline{5-6}
                       & \emph{GFPT-NI}            & \emph{CFPT}          &  \emph{MFPT}          & \emph{CFPT}                         & \emph{MFPT}                        \\
\hline
er\_150\_0100          & 1682.99         & 13.58          & 21.01        & 136.23                       & 97.15                       \\
er\_150\_0102          & 3199.13         & 45.69          & 44.96        & 113.52                       & 92.12                       \\
er\_150\_0104          & 3280.49         & 233.21         & 52.09        & 72.80                        & 105.75                      \\
er\_150\_0106          & 2280.25         & 35.05          & 21.77        & 136.50                       & 129.77                      \\
er\_150\_0108          & 2570.26         & 12.97          & 20.85        & 338.74                       & 166.53                      \\
er\_150\_0110          & 4312.95         & 54.09          & 66.32        & 694.26                       & 96.58                       \\
                       &                 &                &              &                              &                             \\
er\_250\_0070          & 1587.49         & 13.40          & 8.62         & 98.75                        & 113.62                      \\
er\_260\_0070          & 3801.74         & 39.39          & 16.39        & 322.02                       & 170.75                      \\
er\_270\_0070          & 3325.30         & 8.21           & 22.35        & 243.79                       & 128.06                      \\
er\_280\_0070          & 523.17          & 8.52           & 10.85        & 87.87                        & 50.40                       \\
er\_290\_0070          & 2148.48         & 13.09          & 34.77        & 188.97                       & 64.34                       \\
er\_300\_0070          & 1466.10         & 2.14           & 15.53        & 724.81                       & 74.41                       \\
                       &                 &                &              &                              &                             \\
er\_200\_009           & 3086.79         & 16.72          & 18.90        & 468.59                       & 123.09                      \\
er\_500\_005           & 3283.63         & 3.67           & 14.74        & 493.49                       & 172.50                      \\
er\_750\_0040          & 2581.47         & 3.48           & 43.97        & 544.25                       & 123.23                      \\
er\_1000\_0035         & 2324.95         & 4.03           & 18.02        & 298.09                       & 81.38             \\
\hline
\end{tabular}
\end{minipage}
\label{tab:fptRunTime2}
\end{table}
\onecolumn

\subsection{Dataset for the analysis of Section \ref{exp_greedy}.}
\label{sec:app_exp_dataset2}

Table \ref{tab:dataset_brain} provides mappings of brain networks listed in Table \ref{tab:dataset} to their original names in \url{networkrepository.com}.

\begin{table}[H]
\caption{Dataset of real-world graphs. First 21 graphs are brain networks, middle 3 graphs are gene regulatory networks, and last 6 graphs are social networks.}
\begin{minipage}{\textwidth}
\centering
\begin{tabular}{lrrrrr}
\hline
\multicolumn{1}{c}{Graph}      & \multicolumn{1}{c}{$n$ \footnote{Isolated vertices are not counted.}} & \multicolumn{1}{c}{$m$}             &  \multicolumn{1}{c}{$d$} & \multicolumn{1}{>{\centering\arraybackslash}m{23mm}}{$\Delta$} & \multicolumn{1}{>{\centering\arraybackslash}m{23mm}}{Trivial Cliques} \\
\hline
Brain-1         & 177,584       & 15,669,037    & 386        & 5,419                 & 12,051          \\
Brain-2         & 763,149       & 40,258,003    & 308        & 3,951                 & 5,928           \\
Brain-3         & 737,579       & 50,037,313    & 523        & 5,233                 & 5,418           \\
Brain-4         & 835,832       & 59,548,327    & 460        & 4,614                 & 4,508           \\
Brain-5         & 851,113       & 67,658,067    & 483        & 7,272                 & 3,985           \\
Brain-6         & 428,842       & 79,114,771    & 605        & 7,793                 & 4,987           \\
Brain-7         & 935,265       & 87,273,967    & 367        & 9,475                 & 1,611           \\
Brain-8         & 924,284       & 94,370,886    & 414        & 8,135                 & 1,778           \\
Brain-9         & 701,145       & 103,134,404   & 600        & 11,205                & 5,965           \\
Brain-10        & 714,571       & 112,519,748   & 561        & 6,686                 & 5,639           \\
Brain-11        & 791,219       & 121,907,663   & 767        & 4,841                 & 4,754           \\
Brain-12        & 742,862       & 131,926,773   & 754        & 7,685                 & 6,073           \\
Brain-13        & 692,397       & 140,102,158   & 888        & 8,811                 & 5,355           \\
Brain-14        & 727,487       & 150,443,553   & 1,054      & 7,889                 & 5,648           \\
Brain-15        & 748,521       & 159,835,566   & 1,807      & 11,354                & 5,868           \\
Brain-16        & 728,874       & 171,231,873   & 806        & 9,679                 & 4,448           \\
Brain-17        & 774,886       & 181,569,095   & 1,040      & 15,687                & 3,745           \\
Brain-18        & 707,284       & 191,224,983   & 964        & 8,638                 & 2,930           \\
Brain-19        & 776,644       & 201,198,184   & 960        & 10,129                & 3,963           \\
Brain-20        & 753,905       & 209,976,387   & 1,522      & 14,033                & 4,353           \\
Brain-21        & 784,262       & 267,844,669   & 1,200      & 21,743                & 3,672           \\
                &               &               &            &                       &                 \\
human\_gene1    & 21,890        & 12,323,680    & 2,047      & 7,938                 & 945             \\
human\_gene2    & 14,022        & 9,027,024     & 1,902      & 7,228                 & 326             \\
mouse\_gene     & 43,126        & 14,461,095    & 1,045      & 8,031                 & 5,551           \\
                &               &               &            &                       &                 \\
com-DBLP        & 317,080       & 1,049,866     & 113        & 343                   & 73,142          \\
com-Amazon      & 334,863       & 925,872       & 6          & 549                   & 211,212         \\
com-Youtube     & 1,134,890     & 2,987,624     & 51         & 28,754                & 1,590,346       \\
com-LiveJournal & 3,997,962     & 34,681,189    & 360        & 14,815                & 5,809,704       \\
com-Orkut       & 3,072,441     & 117,185,083   & 253        & 33,313                & 15,988,540      \\
com-Friendster  & 65,608,366    & 1,806,067,135 & 304        & 5,214                 & 361,040,101    
 \\
\hline
\end{tabular}
\end{minipage}
\label{tab:dataset}
\end{table}

\clearpage

\subsection{Comparison Results of Greedy Algorithms.}
\label{sec:app_exp_greedyResults}

Table \ref{tab:cover_gALG} and Table \ref{tab:time_gALG} show  show comparisons of \emph{clique cover} sizes and runtimes of our greedy algorithms against \emph{g-ALG}. Table \ref{tab:cover_cALG} and Table \ref{tab:time_cALG} show comparisons of \emph{clique cover} sizes and runtimes of our greedy algorithms against \emph{c-ALG}. 

\begin{table}[H]
\caption{Clique cover size $|\C|$ of \emph{g-ALG} and relative reduction of clique cover size of \emph{CCSG} and \emph{CCSD} (relative to \emph{g-ALG}) for the test graphs shown in Table \ref{tab:dataset}. Out of 30 instances of Table \ref{tab:dataset}, \textit{g-ALG} terminated only on 13 instances within 24 hours.}
\centering
\begin{tabular}{lrrrr}
\hline
\multirow{2}{*}{Graph} & \multirow{2}{*}{m} & \multicolumn{1}{c}{$|\C|$}       & \multicolumn{2}{c}{Relative   Reduction} \\
\cline{4-5}
                       &                    & \emph{g-ALG}     & \emph{CCSG}                & \emph{CCSD}               \\
\hline                       
Brain-1                & 15,669,037         & 1,271,237 & 3.63                & 0.53               \\
Brain-2                & 40,258,003         & 1,356,286 & 33.74               & 38.84              \\
Brain-3                & 50,037,313         & 1,373,963 & 35.70               & 40.10              \\
Brain-4                & 59,548,327         & 1,629,884 & 35.23               & 40.53              \\
Brain-5                & 67,658,067         & 1,704,729 & 36.08               & 41.43              \\
Brain-7                & 87,273,967         & 2,062,624 & 36.07               & 41.82              \\
Brain-8                & 94,370,886         & 2,098,627 & 36.12               & 41.85              \\
Brain-9                & 103,134,404        & 1,801,620 & 44.16               & 48.24              \\
Brain-10               & 112,519,748        & 1,915,655 & 45.08               & 49.08              \\
Brain-11               & 121,907,663        & 2,133,026 & 44.49               & 48.61              \\
com\_DBLP              & 1,049,866          & 238,870   & 0.76                & 0.70               \\
com\_Amazon            & 925,872            & 447,488   & 0.50                & 0.61               \\
com\_Youtube           & 2,987,624          & 2,191,576 & 1.20                & -0.39              \\
\hline
\multicolumn{3}{l}{Average relative   reduction}        & 27.14               & 30.15                 \\
\hline
\end{tabular}
\label{tab:cover_gALG}
\end{table}

\begin{table}[H]
\caption{Runtimes of \emph{g-ALG} and relative performances of \emph{CCSG} and \emph{CCSD} (relative to \textit{g-ALG}) on the test graphs of Table \ref{tab:dataset}. Out of 30 instances of Table \ref{tab:dataset}, \textit{g-ALG} terminated only on 13 instances within 24 hours.}
\centering
\begin{tabular}{lrrrr}
\hline
\multirow{2}{*}{Graph} & \multirow{2}{*}{m} & \multicolumn{1}{c}{Time (seconds)} & \multicolumn{2}{c}{Relative Performance} \\
\cline{4-5}
                       &                    & \multicolumn{1}{c}{\emph{g\_ALG}}         & \emph{CCSG}                & \emph{CCSD}               \\
\hline                       
Brain-1                & 15669037           & 79,834.90      & 848.53              & 936.25             \\
Brain-2                & 40258003           & 2,930.59       & 45.91               & 56.51              \\
Brain-3                & 50,037,313         & 11,065.80      & 123.33              & 151.51             \\
Brain-4                & 59,548,327         & 11,097.10      & 100.28              & 124.69             \\
Brain-5                & 67,658,067         & 12,319.00      & 92.76               & 113.82             \\
Brain-7                & 87,273,967         & 12,987.80      & 69.92               & 89.67              \\
Brain-8                & 94,370,886         & 16,010.10      & 75.78               & 95.30              \\
Brain-9                & 103,134,404        & 55,937.60      & 194.45              & 249.38             \\
Brain-10               & 112,519,748        & 53,891.40      & 175.69              & 224.94             \\
Brain-11               & 121,907,663        & 53,756.10      & 161.83              & 206.78             \\
com\_DBLP              & 1,049,866          & 1.83           & 2.27                & 2.07               \\
com\_Amazon            & 925,872            & 2.02           & 1.79                & 1.76               \\
com\_Youtube           & 2,987,624          & 7,150.57       & 1068.57             & 1073.20            \\
\hline
\multicolumn{3}{l}{Geometric mean of relative performances}                               & 80.22 & 93.91 \\
\hline
\end{tabular}
\label{tab:time_gALG}
\end{table}

\begin{table}[H]
\caption{\emph{Clique cover} size $|\C|$ of \emph{c-ALG} and relative reduction of \emph{clique cover} size of \emph{CCSG} and \emph{CCSD} (relative to \emph{c-ALG}) for the test graphs shown in Table \ref{tab:dataset}. For com\_Friendster, \emph{c-ALG} did not terminate in 24 hours, in the meantime included more than 942.57 million cliques in the \emph{clique cover}}
\centering
\begin{tabular}{lrrrr}
\hline
\multirow{2}{*}{Graph} & \multirow{2}{*}{$m$} & \multicolumn{1}{c}{$|\C|$}                 & \multicolumn{2}{c}{Relative   Reduction} \\
\cline{4-5}
                       &                    & \emph{c-ALG}              & \emph{CCSG}                & \emph{CCSD}               \\
\hline
Brain-1                & 15,669,037         & 1,422,541           & 13.97               & 11.20              \\
Brain-2                & 40,258,003         & 1,019,877           & 11.76               & 18.54              \\
Brain-3                & 50,037,313         & 1,009,000           & 12.32               & 18.32              \\
Brain-4                & 59,548,327         & 1,195,909           & 11.64               & 18.87              \\
Brain-5                & 67,658,067         & 1,230,938           & 11.40               & 18.81              \\
Brain-6                & 79,114,771         & 2,648,702           & 12.95               & 13.92              \\
Brain-7                & 87,273,967         & 1,480,951           & 10.94               & 18.94              \\
Brain-8                & 94,370,886         & 1,504,777           & 10.88               & 18.87              \\
Brain-9                & 103,134,404        & 1,187,423           & 15.13               & 21.34              \\
Brain-10               & 112,519,748        & 1,253,065           & 15.91               & 22.03              \\
Brain-11               & 121,907,663        & 1,393,330           & 14.91               & 21.24              \\
Brain-12               & 131,926,773        & 1,309,727           & 16.51               & 22.09              \\
Brain-13               & 140,102,158        & 1,280,755           & 16.73               & 22.21              \\
Brain-14               & 150,443,553        & 1,328,541           & 16.89               & 22.19              \\
Brain-15               & 159,835,566        & 1,384,057           & 16.74               & 22.31              \\
Brain-16               & 171,231,873        & 1,447,271           & 16.63               & 22.40              \\
Brain-17               & 181,569,095        & 1,555,077           & 15.67               & 22.61              \\
Brain-18               & 191,224,983        & 1,490,147           & 16.64               & 23.04              \\
Brain-19               & 201,198,184        & 1,585,659           & 17.11               & 22.94              \\
Brain-20               & 209,976,387        & 1,553,300           & 16.38               & 22.87              \\
Brain-21               & 267,844,669        & 2,827,079           & 17.47               & 20.80              \\
human\_gene1           & 12,323,680         & 177,210             & 20.38               & 13.65              \\
human\_gene2           & 9,027,024          & 114,916             & 19.48               & 12.79              \\
mouse\_gene            & 14,461,095         & 479,361             & 19.00               & 12.09              \\
com-DBLP               & 1,049,866          & 238,573             & 0.58                & 0.52               \\
com-Amazon             & 925,872            & 448,157             & 0.79                & 0.89               \\
com-Youtube            & 2,987,624          & 2,202,876           & 3.02                & 1.46               \\
com-LiveJournal        & 34,681,189         & 14,137,013          & 4.11                & 3.23               \\
com-Orkut              & 117,185,083        & 45,361,490          & 5.49                & 4.17               \\
\hline
\multicolumn{3}{l}{Average relative   reduction}                  & 13.15               & 16.36              \\
\hline
                       &                    & \multicolumn{3}{l}{Number of cliques (million)}                    \\
\cline{3-5}                 
com-Friendster         & 1,806,067,135      & \textgreater 942.57 & 919.20              & 922.49        \\
\hline
\end{tabular}
\label{tab:cover_cALG}
\end{table}

\begin{table}[H]
\caption{Runtimes of \emph{c-ALG} and relative performances of \emph{CCSG} and \emph{CCSD} (relative to \textit{c-ALG}) on the test graphs of Table \ref{tab:dataset}.}
\label{tab:time_cALG}
\centering
\begin{tabular}{lrrrr}
\hline
\multirow{2}{*}{Graph} & \multirow{2}{*}{m} & \multicolumn{1}{c}{Time (seconds)}  & \multicolumn{2}{c}{Relative Performance} \\
\cline{4-5}
                       &                    & \multicolumn{1}{c}{\emph{c-ALG}}           & \emph{CCSG}                & \emph{CCSD}               \\
\hline                       
Brain-1                & 15,669,037         & 169.53          & 1.80                & 1.99               \\
Brain-2                & 40,258,003         & 170.81          & 2.68                & 3.29               \\
Brain-3                & 50,037,313         & 226.00          & 2.52                & 3.09               \\
Brain-4                & 59,548,327         & 321.15          & 2.90                & 3.61               \\
Brain-5                & 67,658,067         & 338.52          & 2.55                & 3.13               \\
Brain-6                & 79,114,771         & 862.73          & 1.88                & 2.10               \\
Brain-7                & 87,273,967         & 481.30          & 2.59                & 3.32               \\
Brain-8                & 94,370,886         & 557.21          & 2.64                & 3.32               \\
Brain-9                & 103,134,404        & 653.28          & 2.27                & 2.91               \\
Brain-10               & 112,519,748        & 701.67          & 2.29                & 2.93               \\
Brain-11               & 121,907,663        & 795.12          & 2.39                & 3.06               \\
Brain-12               & 131,926,773        & 1065.05         & 2.71                & 3.45               \\
Brain-13               & 140,102,158        & 1059.41         & 2.45                & 3.08               \\
Brain-14               & 150,443,553        & 1261.73         & 2.72                & 3.31               \\
Brain-15               & 159,835,566        & 1478.49         & 2.91                & 3.58               \\
Brain-16               & 171,231,873        & 1517.51         & 2.75                & 3.35               \\
Brain-17               & 181,569,095        & 1511.06         & 2.41                & 3.02               \\
Brain-18               & 191,224,983        & 1767.84         & 2.77                & 3.49               \\
Brain-19               & 201,198,184        & 1968.44         & 3.06                & 3.82               \\
Brain-20               & 209,976,387        & 1974.25         & 2.60                & 3.20               \\
Brain-21               & 267,844,669        & 2996.41         & 2.30                & 2.82               \\
human\_gene1           & 12,323,680         & 494.86          & 1.84                & 1.70               \\
human\_gene2           & 9,027,024          & 352.61          & 1.80                & 1.72               \\
mouse\_gene            & 14,461,095         & 260.90          & 1.95                & 2.00               \\
com\_DBLP              & 1,049,866          & 2.58            & 3.21                & 2.92               \\
com\_Amazon            & 925,872            & 3.50            & 3.11                & 3.05               \\
com\_Youtube           & 2,987,624          & 16.53           & 2.47                & 2.48               \\
com\_LiveJournal       & 34,681,189         & 194.50          & 3.02                & 3.04               \\
com\_Orkut             & 117,185,083        & 1150.58         & 3.11                & 3.33               \\
\hline
\multicolumn{3}{l}{Geometric Mean of relative performances}                            & 2.51                & 2.91               \\
\hline
                       &                    & \multicolumn{3}{c}{Time (hour)}                        \\
\cline{3-5}                       
com\_Friendster        & 1,806,067,135      & \textgreater 24 & 2.57                & 2.26              \\
\hline
\end{tabular}
\end{table}

\clearpage

\subsection{Comparison results of \emph{CCSG} and FPT algorithms.}
\label{sec:app_exp_greedyFPT}
Table \ref{tab:ccsg_mfpt} show comparison results of \emph{CCSG} and \emph{MFPT}, and Table \ref{tab:ccsg_cfpt} show comparison results of \emph{CCSG} and \emph{CFPT}.

\begin{table}[H]
\caption{Average \emph{clique cover} size $|\C|$ and runtime of \emph{CCGS} and \emph{MFPT} on the graph ensembles of Table \ref{tab:datasetFPT}. Second last column shows geometric mean of ratios of \emph{clique cover} size of \emph{CCSG} to \emph{clique cover} size of \emph{MFPT}. Last column shows geometric mean of ratios of runtime of \emph{MFPT} to runtime of \emph{CCSG}.}
\centering
\begin{tabular}{@{\extracolsep{6pt}}lrrrrrr@{}}
\hline
\multirow{2}{*}{Graph} & \multicolumn{2}{c}{$|\C|$ (average)} & \multicolumn{2}{l}{Runtime (average)} & \multicolumn{1}{c}{Geo. mean of ratios} & \multicolumn{1}{c}{Geo. mean of ratios} \\
\cline{2-3} \cline{4-5}
                        & \emph{CCSG}                    & \emph{MFPT}                    & \emph{CCSG}             & \emph{MFPT}               & \multicolumn{1}{c}{of clique cover size}     & \multicolumn{1}{c}{of runtime} \\
\hline                        
er\_150\_0100           & 622.03                  & 611.41                  & 0.001            & 354.19             & 1.017                    & 30,800.67         \\
er\_150\_0102           & 625.03                  & 613.66                  & 0.001            & 569.32             & 1.018                    & 46,283.01         \\
er\_150\_0104           & 627.03                  & 615.65                  & 0.001            & 319.25             & 1.018                    & 30,700.69         \\
er\_150\_0106           & 625.74                  & 613.55                  & 0.001            & 1,037.81           & 1.020                    & 40,508.13         \\
er\_150\_0108           & 629.94                  & 617.72                  & 0.001            & 706.91             & 1.020                    & 52,906.64         \\
er\_150\_0110           & 640.88                  & 627.53                  & 0.001            & 884.55             & 1.021                    & 187,999.12        \\
                        &                         &                         &                  &                    &                          &                   \\
er\_250\_0070           & 1,312.32                & 1,295.90                & 0.002            & 218.06             & 1.013                    & 10,727.51         \\
er\_260\_0070           & 1,401.39                & 1,383.90                & 0.002            & 763.76             & 1.013                    & 16,957.37         \\
er\_270\_0070           & 1,495.16                & 1,476.22                & 0.002            & 393.62             & 1.013                    & 27,719.52         \\
er\_280\_0070           & 1,580.63                & 1,555.97                & 0.002            & 1,503.47           & 1.016                    & 35,949.35         \\
er\_290\_0070           & 1,677.81                & 1,651.78                & 0.002            & 377.25             & 1.016                    & 32,558.44         \\
er\_300\_0070           & 1,775.19                & 1,748.48                & 0.002            & 707.30             & 1.015                    & 27,660.49         \\
                        &                         &                         &                  &                    &                          &                   \\
er\_200\_009            & 973.03                  & 955.81                  & 0.001            & 126.28             & 1.018                    & 21,467.68         \\
er\_500\_005            & 3,732.63                & 3,685.84                & 0.005            & 1,213.04           & 1.013                    & 13,188.51         \\
er\_750\_0040           & 6,822.45                & 6,738.94                & 0.009            & 786.22             & 1.012                    & 5,692.55          \\
er\_1000\_0035          & 10,551.19               & 10,417.03               & 0.015            & 708.18             & 1.013                    & 2,426.77         \\
\hline
\end{tabular}
\label{tab:ccsg_mfpt}
\end{table}

\begin{table}[H]
\caption{Average \emph{clique cover} size $|\C|$ and runtime of \emph{CCGS} and \emph{CFPT} on the graph ensembles of Table \ref{tab:datasetFPT}. Second last column shows geometric mean of ratios of \emph{clique cover} size of \emph{CCSG} to \emph{clique cover} size of \emph{CFPT}. Last column shows geometric mean of ratios of runtime of \emph{CFPT} to runtime of \emph{CCSG}.}
\centering
\begin{tabular}{@{\extracolsep{6pt}}lrrrrrr@{}}
\hline
\multirow{2}{*}{Graph} & \multicolumn{2}{c}{$|\C|$ (average)} & \multicolumn{2}{l}{Runtime (average)} & \multicolumn{1}{c}{Geo. mean of ratios} & \multicolumn{1}{c}{Geo. mean of ratios} \\
\cline{2-3} \cline{4-5}
                        & \emph{CCSG}                    & \emph{CFPT}                    & \emph{CCSG}             & \emph{CFPT}               & \multicolumn{1}{c}{of clique cover size}     & \multicolumn{1}{c}{of runtime} \\
\hline                        
er\_150\_0100           & 622.03                  & 611.32                  & 0.001             & 214.37            & 1.018                    & 19,290.66         \\
er\_150\_0102           & 625.03                  & 613.66                  & 0.001             & 114.41            & 1.019                    & 20,203.29         \\
er\_150\_0104           & 626.77                  & 615.47                  & 0.001             & 475.88            & 1.018                    & 31,971.71         \\
er\_150\_0106           & 625.80                  & 613.87                  & 0.001             & 190.10            & 1.019                    & 27,537.44         \\
er\_150\_0108           & 629.94                  & 617.72                  & 0.001             & 427.45            & 1.020                    & 27,949.53         \\
er\_150\_0110           & 640.42                  & 627.10                  & 0.001             & 295.32            & 1.021                    & 40,520.17         \\
                        &                         &                         &                   &                   &                          &                   \\
er\_250\_0070           & 1,312.97                & 1,296.38                & 0.002             & 248.23            & 1.013                    & 7,433.56          \\
er\_260\_0070           & 1,401.78                & 1,384.13                & 0.002             & 407.72            & 1.013                    & 6,065.27          \\
er\_270\_0070           & 1,495.16                & 1,476.22                & 0.002             & 154.87            & 1.013                    & 7,374.01          \\
er\_280\_0070           & 1,580.63                & 1,555.97                & 0.002             & 252.40            & 1.016                    & 8,149.40          \\
er\_290\_0070           & 1,677.81                & 1,651.78                & 0.002             & 68.65             & 1.016                    & 5,138.18          \\
er\_300\_0070           & 1,775.09                & 1,748.44                & 0.002             & 306.79            & 1.015                    & 4,740.27          \\
                        &                         &                         &                   &                   &                          &                   \\
er\_200\_009            & 973.75                  & 956.25                  & 0.001             & 78.18             & 1.018                    & 8,253.58          \\
er\_500\_005            & 3,732.63                & 3,685.84                & 0.005             & 26.11             & 1.013                    & 5,397.34          \\
er\_750\_0040           & 6,823.72                & 6,740.56                & 0.009             & 63.14             & 1.012                    & 921.63            \\
er\_1000\_0035          & 10,547.88               & 10,414.22               & 0.015             & 49.08             & 1.013                    & 379.11           \\
\hline
\end{tabular}
\label{tab:ccsg_cfpt}
\end{table}

\twocolumn

\subsection{\emph{MFPT} based on different maximal clique enumeration algorithms.}
\label{sec:app_exp_mce}

For maximal clique enumeration, \cite{gramm2006data} used a algorithm proposed by \cite{koch2001enumerating}. We evaluated \emph{MFPT} with three different maximal clique enumeration algorithms: \cite{koch2001enumerating} (\emph{Koch}), \cite{tomita2006worst} (\emph{TTT}), \cite{eppstein2013listing} (\emph{ELS}) and found that implementation of \emph{MFPT} based on \emph{ELS} performs better than implementations based on \emph{Koch} and \emph{TTT}. Our implementation of \emph{ELS} is comparable to \emph{ELS-bare} \cite{eppstein2013listing}, and our implementations of \emph{Koch} and \emph{TTT} are comparable to \emph{TTT-lists} \cite{eppstein2013listing}. Table \ref{tab:mceComp} shows the comparison results of different implementations of \emph{MFPT}. Three of our tested algorithms (\emph{MFPT}, \emph{GFPT}, and \emph{GFPT-NI}) are based on maximal clique enumeration, and for all three algorithms, we have used \emph{ELS}.

\begin{table}[H]
\caption{Average runtime (in seconds) of \emph{MFPT} using different maximal clique enumeration algorithms: \emph{Koch} \cite{koch2001enumerating}, \emph{TTT} \cite{tomita2006worst}, and \emph{ELS} \cite{eppstein2013listing}.}
\begin{tabular}{lrrrr}
\hline
\multirow{2}{*}{Graph} & \multicolumn{3}{l}{Average runtime   (second) of MFPT} \\
\cline{2-4}
                       & Koch             & TTT              & ELS              \\
\hline                       
er\_150\_0100          & 399.74           & 394.22           & 354.19           \\
er\_150\_0102          & 649.06           & 629.29           & 569.32           \\
er\_150\_0104          & 347.66           & 345.65           & 319.25           \\
er\_150\_0106          & 663.60           & 661.24           & 612.69           \\
er\_150\_0108          & 783.42           & 781.28           & 706.91           \\
er\_150\_0110          & 966.23           & 963.55           & 884.55           \\
                       &                  &                  &                  \\
er\_250\_0070          & 236.07           & 232.80           & 218.06           \\
er\_260\_0070          & 844.45           & 840.69           & 763.76           \\
er\_270\_0070          & 426.90           & 415.12           & 393.62           \\
er\_280\_0070          & 1659.90          & 1620.85          & 1503.47          \\
er\_290\_0070          & 407.90           & 405.58           & 377.25           \\
er\_300\_0070          & 788.59           & 772.98           & 707.30           \\
                       &                  &                  &                  \\
er\_200\_009           & 140.98           & 138.34           & 126.28           \\
er\_500\_005           & 1338.58          & 1330.68          & 1213.04          \\
er\_750\_0040          & 413.96           & 408.33           & 379.77           \\
er\_1000\_0035         & 764.38           & 741.85           & 708.18          \\
\hline
\end{tabular}
\label{tab:mceComp}
\end{table}

\newpage

\subsection{Brain Networks.}
\label{sec:app_exp_brain}
Table \ref{tab:dataset_brain}.
\begin{table}[H]
\caption{Mappings of names of the brain networks shown in Table \ref{tab:dataset} to their original names in \url{networkrepository.com} \cite{nr}}
\centering
\begin{tabular}{ll}
\hline
Graph    & Brain Network Name                        \\
\hline
Brain-1  & bn-human-BNU\_1\_0025890\_session\_1    \\
Brain-2  & bn-human-Jung2015\_M87125334            \\
Brain-3  & bn-human-Jung2015\_M87104509            \\
Brain-4  & bn-human-Jung2015\_M87119472            \\
Brain-5  & bn-human-Jung2015\_M87104300            \\
Brain-6  & bn-human-Jung2015\_M87118347            \\
Brain-7  & bn-human-Jung2015\_M87102575            \\
Brain-8  & bn-human-Jung2015\_M87122310            \\
Brain-9  & bn-human-BNU\_1\_0025914\_session\_2    \\
Brain-10 & bn-human-BNU\_1\_0025916\_session\_1    \\
Brain-11 & bn-human-Jung2015\_M87118219            \\
Brain-12 & bn-human-BNU\_1\_0025889\_session\_2    \\
Brain-13 & bn-human-BNU\_1\_0025873\_session\_2-bg \\
Brain-14 & bn-human-BNU\_1\_0025868\_session\_1-bg \\
Brain-15 & bn-human-BNU\_1\_0025918\_session\_1    \\
Brain-16 & bn-human-Jung2015\_M87112427            \\
Brain-17 & bn-human-Jung2015\_M87118465            \\
Brain-18 & bn-human-Jung2015\_M87104201            \\
Brain-19 & bn-human-Jung2015\_M87101705            \\
Brain-20 & bn-human-Jung2015\_M87125286            \\
Brain-21 & bn-human-Jung2015\_M87113878           \\
\hline
\end{tabular}
\label{tab:dataset_brain}
\end{table}
\end{document}